\def\lipics{0}
\def\arxiv{1}

\if\lipics1
\documentclass[a4-paper,USenglish]{lipics-v2019}
\else
\documentclass[11pt]{article}
\fi

\usepackage{amsmath,amsthm,amssymb,amscd,amstext,amsfonts}
\usepackage{mathtools}
\usepackage{mathrsfs}
\usepackage{thmtools}
\usepackage{latexsym}
\usepackage{verbatim}
\usepackage{framed}
\usepackage{graphicx}
\usepackage{stmaryrd}
\usepackage{enumerate}
\if\lipics0
\usepackage{fullpage}
\fi
\usepackage{bm}
\usepackage{color}
\usepackage{hyperref}
\usepackage{url}
\usepackage{physics}

\if\lipics0
\newtheorem{theorem}{Theorem}[section]

\newtheorem{corollary}[theorem]{Corollary}
\newtheorem{lemma}[theorem]{Lemma}

\newtheorem{proposition}[theorem]{Proposition}
\newtheorem{definition}[theorem]{Definition}
\newtheorem{example}[theorem]{Example}

\theoremstyle{remark}
\newtheorem*{remark}{Remark}
\fi
\newtheorem{conjecture}[theorem]{Conjecture}

\DeclareMathOperator{\Ex}{\mathbb{E}}           
\DeclareMathOperator*{\Exs}{\mathbb{E}}           

\usepackage{color,wasysym}

\providecommand{\bin}{\mathrm{bin}}

\if\lipics1
\providecommand{\myparagraph}[1]{\subparagraph{#1}}
\providecommand{\rsz}{lem:single-unbiased}
\else
\providecommand{\myparagraph}[1]{\paragraph{#1}}
\providecommand{\rsz}{rsz}
\fi

\if\lipics1
\title{Biasing Boolean Functions and Collective Coin-Flipping Protocols over Arbitrary Product Distributions}
\titlerunning{Biasing Boolean Functions}

\author{Yuval Filmus}{Computer Science Department, Technion, Haifa, Israel \and \url{http://www.cs.toronto.edu/~yuvalf/}}{yuvalfi@cs.technion.ac.il}{https://orcid.org/0000-0002-1739-0872}{Taub Fellow --- supported by the Taub Foundations. The research was funded by ISF grant 1337/16.}
\author{Lianna Hambardzumyan}{School of Computer Science, McGill University, Montreal, QC, Canada}{lianna.hambardzumyan@mail.mcgill.ca}{}{}
\author{Hamed Hatami}{School of Computer Science, McGill University, Montreal, QC, Canada \and \url{https://www.cs.mcgill.ca/~hatami/}}{hatami@cs.mcgill.ca}{https://orcid.org/0000-0002-4732-434X}{}
\author{Pooya Hatami}{Department of Computer Science, UT Austin, Austin, TX, USA \and \url{https://pooyahatami.org/}}{pooyahat@gmail.com}{https://orcid.org/0000-0001-6361-7604}{}
\author{David Zuckerman}{Department of Computer Science, UT Austin, Austin, TX, USA \and \url{http://www.cs.utexas.edu/~diz/}}{diz@cs.utexas.edu}{}{}

\authorrunning{Y. Filmus, L. Hambardzumyan, H. Hatami, P. Hatami and D. Zuckerman}
\Copyright{Yuval Filmus and Lianna Hambardzumyan and Hamed Hatami and Pooya Hatami and David Zuckerman}

\ccsdesc[300]{Theory of Computation}
\keywords{Boolean function analysis, coin flipping}
\category{}
\relatedversion{}
\supplement{}
\acknowledgements{Part of the work on this paper was done while the first three authors were at the Simons Institute for the Theory of Computing at Berkeley, CA, USA.}

\else
\if\arxiv1
\title{\LARGE{Biasing Boolean Functions and Collective Coin-Flipping Protocols over Arbitrary Product Distributions}}
\else
\title{{\Huge \textbf{Appendix}} \\ Biasing Boolean Functions and Collective Coin-Flipping Protocols over Arbitrary Product Distributions}
\fi
\author{\Large{Yuval Filmus}\thanks{Taub Fellow --- supported by the Taub Foundations. The research was funded by ISF grant 1337/16.}\\ \small{Computer Science Department,} \\ \small{Technion} \\ \small{\texttt{yuvalfi@cs.technion.ac.il}} \and \Large{Lianna Hambardzumyan}\\ \small{School of Computer Science,}\\ \small{McGill University}\\ \small{\texttt{lianna.hambardzumyan@mail.mcgill.ca}} \and \Large{Hamed Hatami}\thanks{Supported by an NSERC grant.} \\ \small{School of Computer Science,}\\ \small{McGill University}\\ \small{\texttt{hatami@cs.mcgill.ca}} \and \Large{Pooya Hatami}\thanks{Supported by a Simons Investigator Award (\#409864, David Zuckerman)}\\ \small{Department of Computer Science,}\\ \small{UT Austin}\\ \small{ \texttt{pooyahat@gmail.com}} \and \Large{David Zuckerman}\thanks{Supported by NSF Grant CCF-1705028 and a Simons Investigator
Award (\#409864)}\\ \small{Department of Computer Science,}\\ \small{UT Austin}\\ \small{\texttt{diz@cs.utexas.edu}}}
\fi

\begin{document}

\maketitle

\begin{abstract}
The seminal result of Kahn, Kalai and Linial shows that a coalition of $O(\frac{n}{\log n})$ players can bias the outcome of \emph{any} Boolean function $\{0,1\}^n \to \{0,1\}$ with respect to the uniform measure. We extend their result to arbitrary product measures on $\{0,1\}^n$, by combining their argument with a completely different argument that handles very biased coordinates. 

We view this result as a step towards proving a conjecture of Friedgut, which states that Boolean functions on the continuous cube $[0,1]^n$ (or, equivalently, on $\{1,\dots,n\}^n$) can be biased using coalitions of $o(n)$ players. This is the first step taken in this direction since Friedgut proposed the conjecture in 2004. 

\smallskip

Russell, Saks and Zuckerman extended the result of Kahn, Kalai and Linial to multi-round protocols, showing that when the number of rounds is $o(\log^* n)$, a coalition of $o(n)$ players can bias the outcome with respect to the uniform measure. We extend this result as well to arbitrary product measures on $\{0,1\}^n$.

The argument of Russell et al.\ relies on the fact that a coalition of $o(n)$ players can boost the expectation of any Boolean function from $\epsilon$ to $1-\epsilon$ with respect to the uniform measure. This fails for general product distributions, as the example of the AND function with respect to $\mu_{1-1/n}$ shows. Instead, we use a novel boosting argument alongside a generalization of our first result to arbitrary finite ranges.
\end{abstract}

\section{Introduction}
How can distributed processors collectively flip a somewhat fair coin if some processors may try to bias the outcome? 
In the \emph{Collective Coin-Flipping Problem}, a classical problem in distributed computing, $n$ processors wish to generate a single common random bit, even in the presence of faulty and possibly malicious processors. Collective coin-flipping protocols can be used to expedite \emph{Byzantine Agreement}~\cite{CD89} and are closely related to \emph{Leader Election Protocols}~\cite{Dodis06}. The problem has been considered in several scenarios, depending on the assumptions made on the type of the communication between the processors, the kind and number of faults, and the power of the adversary~\cite{CD89, BLS89, Dodis06, BL89}. 

A Boolean function $f\colon \{0,1\}^n\rightarrow \{0,1\}$, where $\{0,1\}^n$ is endowed with a product measure $\mu$, naturally corresponds to a \emph{single round collective coin-flipping protocol} in the \emph{perfect information} model introduced by Ben-Or and Linial~\cite{BL89}, where $n$ players each broadcast a bit according to a private distribution, and at the end, the output of the protocol is the value of $f$ on the broadcast string. An interesting and important concept in the design of collective coin-flipping protocols is \emph{resilience} against coalitions of a significant number of players who attempt to influence the output of the protocol towards a particular value.

A \emph{coalition} is a subset $S$ of players that have a particular desired value $b \in \{0,1\}$ in mind, and if possible, broadcast bits that set the output of the protocol to $b$.
We study the model where the coalition is allowed \emph{rushing}: the corrupt players may wait until all the other players broadcast their bits before deciding on what bit to broadcast. In other words, they \emph{succeed} on $x \sim \mu$ if it is possible to modify $x$ only on the coordinates in $S$ to obtain a string $y$ with $f(y)=b$; they \emph{fail} if the value of $f$ is already determined to be \emph{not} equal to $b$ by the bits broadcast by the players outside the coalition. The success of such a coalition can be easily quantified as the probability that the coalition succeeds on a random $x \sim \mu$. 

 Fix a parameter $\epsilon>0$. A protocol $f$ is said to be \emph{$\epsilon$-resilient against coalitions of $\ell$ players}  if no coalition of size at most $\ell$ succeeds with probability at least  $1-\epsilon$. How resilient can a function be against large coalitions? Over the uniform distribution, perhaps  the most natural candidate for a highly resilient function is the majority function, which can be easily seen to be resilient against $\Omega(\sqrt{n})$ size coalitions. However, somewhat surprisingly, it turns out that plain democracy is not the most effective way to be immune against the influence of  coalitions. Indeed, Ajtai and Linial~\cite{AL93} gave a \emph{randomized} construction of a Boolean function that is resilient against coalitions of size $\Omega(n/\log^2 n)$, significantly better than the $\Omega(\sqrt{n})$ bound of the majority function.  More recently, Chattopadhyay and Zuckerman~\cite{CZ16} gave an \emph{explicit} construction of a highly resilient function over the uniform measure.  This was a key ingredient in their breakthrough work that introduced explicit two-source extractors for polylogarithmic min-entropy. Subsequently, Meka~\cite{Meka17} gave an explicit construction of a monotone depth three Boolean function that is as resilient as the randomized construction of Ajtai and Linial. 
 
In this  article, we are mainly interested in the limitations of resilience.  The most classical theorem in this direction is  due to Kahn, Kalai, and Linial~\cite{KKL}, who proved  that, for the uniform distribution, no Boolean function is resilient against coalitions of size $\omega(n/\log n)$.  Closing the  gap between this bound and  the $\Omega(n/\log^2 n)$ construction of  Ajtai and Linial remains a longstanding open problem. 

Starting with the work of Ben-Or and Linial~\cite{BL89}, researchers have studied two natural ways to generalize the discussed protocols: First, allow players to broadcast longer messages, and second, allow many rounds. In this paper, we mostly focus on the latter generalization. 
In the multi-round setting, the voting procedure that is described above is repeated $r$ times: at every round, first the players who are not in the coalition broadcast their random messages, and then  the players in the coalition decide and broadcast their messages in an adversarial manner. When the players are sending single-bit messages, the outcome is decided by a function $f\colon (\{0,1\}^n)^r \to \{0,1\}$. 

The most efficient known protocols are due to Russell and Zuckerman~\cite{RZ01} and to Feige~\cite{Feige99}. In the case where players are allowed to send longer messages, they constructed $\log^\star n+ O(1)$ round protocols resilient against coalitions of size $\beta n$ for any $\beta<1/2$. In the case when players are allowed to broadcast single bit messages, their protocols use $(1-o(1))\log n$ rounds, and are still resilient against coalitions of size $\beta n$ for any $\beta<1/2$. For a discussion of various models and known upper and lower bounds, see a survey of Dodis~\cite{Dodis06}.

In the multi-round setting, the players in the  coalition have the disadvantage that they will not see the future-round votes of the other players before voting in the current round. Thus, it becomes significantly more difficult to prove limitations on resilience as $r$ grows, and naturally the known bounds are weaker.
  Russell, Saks and Zuckerman~\cite{RSZ02}, building upon the work of Kahn et al.~\cite{KKL}, showed that over the uniform measure, no Boolean function $f\colon (\{0,1\}^n)^r \to \{0,1\}$ is $\epsilon$-resilient against coalitions of size $\omega_\epsilon\left(\frac{r^2 n}{\log^{(2r-1)} n}\right)$, where $\log^{(2r-1)} n$ is an iterated logarithm. It follows as a simple corollary that $\Omega(\log^\star n)$ rounds are necessary in order for a protocol to be resilient against coalitions of size $\Omega(n)$. 
 
\medskip

The purpose of this paper is to generalize the above results from the uniform distribution to arbitrary product distributions on the Boolean cube. 
\medskip

A moment of reflection reveals that there are  major differences between the uniform distribution and the general case, and  
 indeed, prior to this work, it was not clear to us whether similar results were true for general product distributions. We will elaborate on this  later, but for now, we only mention that the coordinates $x_i$ that are not highly biased, i.e. $t \le \Pr[x_i=1] \le 1-t$ for some $t$ that is not too small, can be handled using the same argument as in Kahn et al.~\cite{KKL}. Similarly, the argument of  Russell et al.~\cite{RSZ02} can be used to analyze these coordinates in the multi-round setting. However, the highly biased coordinates behave very differently, and to handle those, we need to take an entirely new  approach, and employ  a new set of ideas. Indeed, our proofs for the highly biased case have almost no resemblance to those in previous works. 
 
 Our first theorem concerns  single round protocols. By combining the argument of Kahn, Kalai and Linial with an argument geared towards biased coordinates, we are able to show that these protocols can always be influenced towards a single value, with coalitions which are only slightly worse than those guaranteed by the KKL theorem.

\begin{theorem}
\label{thm:singleRoundSimplified}
Over any product distribution  $\mu$,  there is no  function $f\colon \{0,1\}^n  \to \{0,1\}$ that is $\epsilon$-resilient against coalitions of size $\omega_\epsilon(\frac{n\log \log n}{\log n})$ . 
\end{theorem}
(In contrast, the KKL theorem shows the impossibility of $\epsilon$-resilience against coalitions of size $\omega_\epsilon(\frac{n}{\log n})$.)

Next, we prove an impossibility result for resilience in the  multi-round setting over arbitrary product distributions. This was posed as an open problem by Russell et al.~\cite{RSZ02}. Here we face several new challenges.  Generalizing our argument for the biased  coordinates to the multi-round setting is far from straightforward, and combining it with the argument of Russell et al.~\cite{RSZ02} for the unbiased coordinates also requires new ideas.  

\begin{theorem}
\label{thm:multiRoundSimplified}
Let $n$ and $r<n$ be given. Over any product distribution  $\mu$ over $(\{0,1\}^n)^r$, there is no $r$-round coin-flipping protocol $f\colon (\{0,1\}^n)^r  \to \{0,1\}$ that is $\epsilon$-resilient against coalitions of size $\omega_\epsilon\left(\frac{n (\log^\star n)^2 }{\log^{(4r)}n}\right)$.
\end{theorem}
As a result, over any product distribution $\mu$, $\Omega(\log^\star n)$ rounds are necessary in order for a protocol to be resilient against coalitions of size $\Omega(n)$.

\myparagraph{Influences} 

The notion of resilience of a Boolean function is related to the influences of variables and coalitions of variables. For a Boolean function $f\colon \{0,1\}^n\rightarrow \{0,1\}$ over a  product probability measure $\mu$,  the \emph{influence} of the $k$-th variable is defined as 
\[I_k(f):=\Pr_{x\sim \mu}\left[ f \text{ is not constant on } B_k(x)\right], \]
where 
\if\lipics1
$B_k(x):= \{y\in \{0,1\}^n \colon y_j = x_j \text{ for all } j \neq k \}$.

\else
\[ B_k(x):= \{y\in \{0,1\}^n \colon y_j = x_j \text{ for all } j \neq k \}. \]
\fi 
The influence of the $k$-th variable \emph{towards a value $b \in \{0,1\}$} is defined as 
\[I_k^b(f):=\Pr_{x\sim \mu}\left[b \in f(B_k(x))\right]. \]
Similarly, the \emph{influence of a coalition} $S\subseteq [n]$  towards a value $b \in \{0,1\}$ is defined  as 
\[I_S^b(f):=\Pr_{x\sim \mu}\left[ b \in f(B_S(x))\right], \]
where 
\if\lipics1
$B_S(x):= \{y\in \{0,1\}^n \colon y_j = x_j \text{ for all } j \notin S \}$.

\else
\[B_S(x):= \{y\in \{0,1\}^n \colon y_j = x_j \text{ for all } j \notin S \}. \] 
\fi
Equivalently, $I_S^b(f)$ is the probability that a  random $x \sim \mu$ can be modified on its $S$ variables such that  the output of $f$ becomes $b$.   

A function $f$ is \emph{not} $\epsilon$-resilient against coalitions of size $\ell$ if and only if there exists a  set $S$ of size at most $\ell$ and a value $b$ such that  $I^b_S(f) \ge 1-\epsilon$.   

\medskip

The seminal work of Kahn, Kalai and Linial introduced discrete Fourier-analytic techniques to the study of influences. Their main theorem, known as the KKL inequality, states that over the uniform measure, every unbiased Boolean function $f\colon \{0,1\}^n \rightarrow \{0,1\}$ has an influential variable. Formally, there exists $k$ such that $I_k(f)\geq \Omega(\frac{\alpha \log n}{n})$ when $\alpha \le \Ex[f(x)] \le 1-\alpha$. Let $b \in \{0,1\}$ satisfy $\Pr[f(x)=b] \ge \epsilon$. Then repeated applications of the KKL inequality imply the existence of a set $S$ with $|S|=O_\epsilon \left(\frac{n}{ \log n}\right)$ such that $I_S^b(f)\geq 1-\epsilon$. In particular, there are no   $\omega_\epsilon(n/\log n)$-resilient functions over the uniform distribution. 

The above argument shows that unless $f$ is already very biased towards $0$ or $1$, one can pick any $b \in \{0,1\}$ and find a small coalition $S$ that can bias $f$ towards $b$.  However, this is no longer true if we consider general product distributions. 

\begin{example}
\label{ex:pbiased}
Consider the $p$-biased distribution $\mu_p^n$ over $\{0,1\}^n$, i.e.\ each coordinate is $1$ with probability $p$. Set $p=1/n$ and let  $f$ be the OR function $\bigvee_{i=1}^n x_i$. Obviously, $\Ex[f] = 1-(1-p)^n \approx 1-\frac{1}{e}$, and yet for every $S$ with $|S|=o(n)$, we have $ I_S^0(f) = 1-(1-p)^{n-|S|} \approx 1-\frac{1}{e}$. In other words, despite the fact that the expected value of the function is bounded away from both $0$ and $1$, no small coalition can influence the output of the function towards $0$. However, this is not a counterexample to Theorem~\ref{thm:singleRoundSimplified} because any set $S$ with $|S|=1$ satisfies $I_S^1(f)=1$, and thus the function is not even $1$-resilient.  
\end{example}

As the above example illustrates, part of the difficulty of generalizing the coalition theorem of KKL is to figure out which $b\in \{0,1\}$ to bias towards. 

Using the notation $I^b_S(f)$,  Theorem~\ref{thm:singleRoundSimplified} can be restated as follows. 
\begin{theorem}[Theorem~\ref{thm:singleRoundSimplified} reformulated]
\label{thm:singleRoundReformulated}
Let  $f\colon \{0,1\}^n  \to \{0,1\}$ be a function over a product distribution  $\mu$. There exists a set $S$ of size $O_\epsilon(\frac{n\log \log n}{\log n})$ such that $I_S^b(f) \ge 1- \epsilon$ for some $b \in \{0,1\}$. 
\end{theorem}
\begin{remark}
To simplify the statement, in Theorem~\ref{thm:singleRoundReformulated}, we did not explicitly state the dependence of $|S|$ on $\epsilon$. Our proof  yields the bound  $|S| = O(\frac{\log(1/\epsilon)n}{\epsilon\log n}+\frac{n \log\log n}{\epsilon \log n})$.
\end{remark}

\myparagraph{Continuous cube and a conjecture of Friedgut}
The Bernoulli distribution on $\{0,1\}$ with parameter $p$ can be embedded in the continuous interval $[0,1]$ via the measure-preserving map $\sigma\colon [0,1] \to \{0,1\}$ defined as $\sigma(x)=1$ if and only if $x \ge 1-p$.  By taking the product of these maps,  for every product probability measure $\mu$ on $\{0,1\}^n$, we obtain a measure-preserving map $\sigma_\mu\colon [0,1]^n \to \{0,1\}^n$. As a result, every function $f\colon (\{0,1\}^n,\mu) \to \{0,1\}$ naturally corresponds to  a function $\overline{f}\colon [0,1]^n \to \{0,1\}$  defined by $\overline{f}=f \circ \sigma_\mu$. Note that 
\if\lipics1
$I_S^b[\overline{f}] = I_S^b[f]$,
\else
\[ I_S^b[\overline{f}] = I_S^b[f], \]
\fi
for every $S \subseteq [n]$ and $b \in \{0,1\}$. Thus, a  more general setting for studying resilience is the set of  measurable functions $f\colon [0,1]^n \to \{0,1\}$. Indeed, Bourgain et al~\cite{BKKKL} proved a generalization of the KKL inequality, but erroneously claimed that as a corollary, if $\epsilon \le \Ex[f]$, then $I_S^1[f] \ge 1- \epsilon$ for a set $S$ of size $|S|=o_\epsilon(n)$. Interestingly, Example~\ref{ex:pbiased}, which was introduced  in the same paper to demonstrate that the proof of the KKL inequality breaks down for the continuous  cube,   is also a counterexample to this false claim.  Friedgut~\cite{F04} pointed out this error, and suggested the following tantalizing conjecture to replace the false statement\footnote{Nati Linial told the last author about this error and conjecture years earlier, but as far as we know this is the first published account.}.   
\begin{conjecture}[\cite{F04}]
\label{conj:Friedgut}
Let  $f\colon [0,1]^n \to \{0,1\}$ be a measurable function. There exists a set $S$ of size $o_\epsilon(n)$ such that $I_S^b(f) \ge 1- \epsilon$ for some $b \in \{0,1\}$. 
\end{conjecture}
A standard compression argument shows that it suffices to prove this conjecture for increasing functions, and indeed the original form of the conjecture is stated for increasing functions. Furthermore, by  discretization, the statement can be further reduced to functions  $f\colon \{1,\ldots,n^2\}^n \to \{0,1\}$, where the domain is endowed with the uniform measure. Note that this form of the conjecture corresponds to resilience of one-round collective coin-flipping protocols where each player is allowed to send $\log n$-bit messages.   

The above discussion show that, qualitatively,  Conjecture~\ref{conj:Friedgut} is a generalization of   Theorem~\ref{thm:singleRoundReformulated}, and thus our theorem can be considered as a step towards resolving Friedgut's conjecture. However, our techniques and ideas seem to fall short of proving the full conjecture. 

\myparagraph{Beyond the Boolean range}

As we discussed above, the coalition theorem of KKL says that if $\Ex[f(x)=b] \ge \epsilon$ then there exists a small coalition $S$ such that $I_S^b(f)\geq 1-\epsilon$. Now consider a function $h\colon\{0,1\}^n \to \mathcal{R}$ over the uniform distribution, where $\mathcal{R}$ is a constant size  set. Pick any $b \in \mathcal{R}$ with $\Pr[h(x)=b] \ge \epsilon$. We can apply the KKL theorem to the function $f\colon \{0,1\}^n \to \{0,1\}$ defined as $f(x)=1$ if and only if $h(x)=b$, and  conclude that there is a  coalition of size  $O_\epsilon \left(n/ \log n\right)$ with  $I_S^b(f)\geq 1-\epsilon$. This shows that over the uniform distribution, the general range $\mathcal{R}$ easily reduces  to the Boolean range. 

Unfortunately, the above reduction cannot be carried  for  general product distributions, for   in  Theorem~\ref{thm:singleRoundReformulated}, the final outcome $b$ is dictated to us by the function. To illustrate the problem,  consider  a function $h\colon \{0,1\}^n \to \{0,1,2\}$ and a general product distribution $\mu$. By bundling $\{1,2\}$ into a single value and applying Theorem~\ref{thm:singleRoundReformulated}, we can conclude that there exists a small coalition $S$ such that  either  it  biases the outcome of the function towards $0$, or it biases the outcome towards being in $\{1,2\}$. If  it is the former  case, then we are done, but in the latter case, it is not clear how to proceed.

 We know that except for the $x$'s that belong to  a small-measure set $\mathcal{E}$,  the coalition can modify $x$  in such a way that the outcome is in $\{1,2\}$. Now at first glance, it might seem that by applying  Theorem~\ref{thm:singleRoundReformulated}  again, we can   find another coalition $T$ that can modify $x$ further to refine the outcome to a single value $b \in \{1,2\}$, and thus conclude that for most $x$'s the alliance $S \cup T$   can influence the outcome of the function towards $b$.
Unfortunately, this is actually not the case. One reason is that $S$ and $T$ might intersect, and suggest conflicting modifications to $x$. Even if $S$ and $T$ are disjoint, the proof doesn't work: denoting by $x'$ the vector obtained from $x \sim \mu$ after modification by $S$, we no longer have $x' \sim \mu$, and so there is no guarantee that on most inputs $T$ can be applied successfully. In other words, $\Pr[x' \in \mathcal{E}]$ need not be small.


The above discussion shows that one cannot deduce the general case via the simple reduction that was outlined above for the uniform measure,  but  surely, as cumbersome  as it may be,   one can go over the proof and generalize every step  from  $\{0,1\}$  to $\{0,1,2\}$ by making small notational adjustments.  This turns out not to be the case  either! The proof of  Theorem~\ref{thm:singleRoundReformulated},  rather unexpectedly, relies on the assumption that the function  takes only two values. Indeed, to generalize the result to larger ranges, we had to introduce new ideas, and in particular a strengthening of Theorem~\ref{thm:singleRoundReformulated} (see Theorem~\ref{thm:strong} below) that provides stronger control over the set  $\mathcal{E}$ described above.

\begin{theorem}[Single round, general range]
\label{thm:singleRound_General}
Let $\mathcal{R}$ be a constant size set, and  $f\colon \{0,1\}^n  \to \mathcal{R}$ be a function over a product distribution  $\mu$. There exists a set $S$ of size $O_\epsilon(\frac{n\log \log n}{\log n})$ such that $I_S^b(f) \ge 1- \epsilon$ for some $b \in \mathcal{R}$. 
\end{theorem}

\begin{remark}
At the heart of the proof of Theorem~\ref{thm:singleRound_General} there is an intermediate result, Theorem~\ref{thm:strong}, which states that if all coordinates are biased, say $\Pr[x_i=1] < \alpha$, then a random coalition of size $O(\log^3 |\mathcal{R}| \log\log |\mathcal{R}| \cdot \alpha n)$ biases the outcome with high probability. This intermediate result is an essential ingredient in the proof of our result on the multi-round setting,  Theorem~\ref{thm:multiRoundSimplified}.
For this application, it was crucial to obtain a bound which depends only polylogarithmically  in $|\mathcal{R}|$.
\end{remark}


Even though Theorem~\ref{thm:singleRoundReformulated} is a special case of   Theorem~\ref{thm:singleRound_General}, we prove them separately, as Theorem~\ref{thm:singleRoundReformulated}  can be proven using a shorter and simpler proof. 

\myparagraph{Paper organization} We prove Theorem~\ref{thm:singleRoundSimplified}, which shows that all single-round protocols can be biased using coalitions of size $o(n)$, in Section~\ref{sec:single-round}. We prove Theorem~\ref{thm:singleRound_General}, which generalizes the preceding result to arbitrary finite domains, in Section~\ref{sec:largerange}. We prove our main result, Theorem~\ref{thm:multiRoundSimplified}, which shows the multi-round protocols can be biased, in Section~\ref{sec:multi-round}.
\if\lipics0
Finally, Section~\ref{sec:conclusion} presents some concluding remarks.
\fi

\if\lipics1
Due to space constraints, some of the proofs are available only in the full version of the paper, which is attached as an appendix.
\fi

\section{Single Round Case: Proof of Theorem~\ref{thm:singleRoundSimplified}} \label{sec:single-round}

In this section we prove Theorem~\ref{thm:singleRoundSimplified}, showing that, under any product distribution, there exists  a small coalition which can bias the output of the function towards one of the outputs.

Note that in order to prove  Theorem~\ref{thm:singleRoundSimplified}, without loss of generality, we can assume that  $\Pr_{x \sim \mu}[x_i=1] \le \frac{1}{2}$  for every $i \in [n]$, as otherwise we can  simply change the role of $0$ and $1$ for the $i$-th coordinate. In light of this observation, the coordinates can be divided into two sets: the small bias coordinates,  satisfying $\Pr_{x\sim \mu}[x_i=1] \in (\alpha_0,\frac{1}{2}]$, and the highly biased coordinates,  satisfying $\Pr_{x\sim \mu}[x_i=1] \le \alpha_0$, where $\alpha_0$  is a threshold that is chosen to be $\alpha_0=\frac{1}{\log n}$.  
 
 Indeed, we first consider the case where all the coordinates are of the same type:  
\begin{itemize}
\item \emph{Small bias case}: $\Pr_{x\sim \mu}[x_i=1] \in (\alpha_0,\frac{1}{2}]$ for every $i \in [n]$.
\item \emph{Large bias case}: $\Pr_{x\sim \mu}[x_i=1] \le \alpha_0$ for every $i \in [n]$.
\end{itemize}

We handle the large bias case in Section~\ref{sec:singleround_largebias}, which is the novel part of the proof. The  small bias case is  handled in  Section~\ref{sec:singleround_smallbias}  via a reduction to the previous work of Russell et al.~\cite{RSZ02}. Finally, in Section~\ref{sec:combine} we show how to combine the two cases to handle any product distribution $\mu$, thus completing the proof of Theorem~\ref{thm:singleRoundSimplified}.

\subsection{Large Bias Case}\label{sec:singleround_largebias}

 We will sometimes identify the subsets of $[n]$ with elements of $\{0,1\}^n$. For example, $S \sim \mu$ would mean that $S= \mathrm{supp}(x)$, where $x$ is sampled according to $\mu$. We construct the coalitions from a certain  boosted form of $\mu$. 

\begin{definition}[Boosted distribution]\label{def:boosted}
For a positive integer $t$, we denote by $\mu^{(t)}$ the distribution of $x^1 \vee \dots \lor x^t$, where  $x^1,\ldots,x^t$ are i.i.d.\ random variables distributed according to $\mu$.
\end{definition}

The large bias case of Theorem~\ref{thm:singleRoundSimplified} follows from the following  general proposition, that holds for distributions that are not necessarily  product distributions. 

\begin{proposition}\label{prop:single-biased}
Consider $f\colon (\{0,1\}^n,\mu) \to \{0,1\}$, where $\mu$ is an \emph{arbitrary} probability measure, and let $S \sim \mu^{(k)}$, where $k \approx \frac{10 \log \frac{1}{\epsilon}}{\epsilon} $.  For some $b \in \{0,1\}$,  we have
$\Pr_S[I^b_S[f] > 1 -\epsilon] > 1- \epsilon$. 
\end{proposition}

Note that Proposition~\ref{prop:single-biased} implies (via a straightforward concentration bound) that in the large bias case, there exists a random coalition of expected size at most $k \alpha_0 n$ such that $\Pr_S[I^b_S[f] > 1 -\epsilon] > 1- \epsilon$. As it will become apparent later, for the application  to the multiround setting, it is important that in Proposition~\ref{prop:single-biased} the set $S$ is   chosen randomly from a distribution that does not depend on $f$. 

Proposition~\ref{prop:single-biased} is a direct consequence of the following lemma, as for the Boolean range $\{0,1\}$, either Condition~I holds for $b=0$ or Condition~II holds for $b=1$. This, however, is not true for larger $\mathcal{R}$.

\begin{lemma}[Key Lemma for Single Round]\label{lem:single-biased}
Consider $f\colon (\{0,1\}^n,\mu) \to \mathcal{R}$, where $\mu$ is an \emph{arbitrary} probability measure. Let $x,y \sim \mu$, $S \sim \mu^{(k)}$, where $k \approx \frac{10 \log \frac{1}{\epsilon}}{\epsilon} $. For $b \in \mathcal{R}$, either of  
\begin{itemize}
\item Condition I:   $\Pr_x[\Pr_y[f(x \vee y)=b] \ge1- \epsilon] > \epsilon/2$, or 
\item Condition II:   $\Pr_x[\Pr_y[f(x \vee y)=b] \ge \epsilon] \ge 1 - \epsilon/2$,
\end{itemize}
implies 
$\Pr_S[I^b_S[f] > 1 -\epsilon] > 1- \epsilon$. 
\end{lemma}

\begin{proof}
Let $S=\mathrm{supp}(y^1 \vee \dots \vee y^k)$, where  $y^1,\ldots,y^k \sim \mu$ are drawn independently. 
Let the sets $X^{I}$ and $X^{II}$  denote the following subsets of the input space $\{0,1\}^n$:
\begin{align*}
 X^{I} &= \{x : \Pr_y[f(x \vee y)=b] \ge1- \epsilon \}, \\
 X^{II} &= \{x : \Pr_y[f(x \vee y)=b] \ge \epsilon\}.
\end{align*}

If we are in the Type I setting, then $\Pr[X^I] > \epsilon/2$, and so  
\[
\Pr_{S}[S \text { contains some } x \in X^I] \ge 1- \Pr[ y^1,\ldots,y^k \not\in X^I] \geq 1-\left(1-\frac{\epsilon}{2}\right)^k> 1-\epsilon.
\]
Note that if there exists $z \in X^I$ which is a subset of $S$ then for every $x$, the two elements $x$ and $x \vee z$ can only differ  on a subset of $S$, and thus 
\[I^{b}_S(f) \ge  \Pr_x[f(x \vee z)=b] > 1-\epsilon.\] 

Now we turn our attention to  Condition II.  In this case, we shall prove that $ \Pr_S[I_S^b[f]<1-\epsilon] \leq \epsilon$. Indeed,
\begin{align}
  \Pr_S[I_S^b[f]<1-\epsilon] &\le \Pr_{y^1,\ldots,y^k}\left[\Pr_x[ \exists i \in [k], \; f(x \vee y^i)=b] < 1-\epsilon \right] \nonumber \\
   &= \Pr_{y^1, \ldots, y^k} \left[\Pr_x[\forall i \in [k], \; f(x \vee y^i) \neq b] \geq \epsilon\right]. \label{eq:pr_to_1}
\end{align}
To bound the last probability, for $x \in \{0,1\}^n$ let $E_x$ denote the  event that  for every $i \in [k]$, $f(x \vee y^i)\neq b$. Then
\[ \Pr_x[E_x]
\leq \Pr_x[x \not\in X^{II}] + \Pr_x[E_x \; \wedge \; x \in X^{II}] \leq \frac{\epsilon}{2} + \Pr_x[E_x \mid x \in X^{II}].\]
Plugging this into~\eqref{eq:pr_to_1}, we get
\begin{multline*}
 \Pr_S[I_S^b[f]<1-\epsilon]  \le \Pr_{y^1, \ldots, y^k} [\Pr_x[E_x] \geq \epsilon] \leq \Pr_{y^1,\ldots,y^k}\Big[\Pr_x [E_x \mid x \in X^{II}] \geq \frac{\epsilon}{2} \Big] \le \\\frac{1}{(\epsilon/2)} \Pr_{y^1, \ldots, y^k,x} [E_x \mid x \in X^{II}] .
 \end{multline*}
Since  $k \approx \frac{10 \log \frac{1}{\epsilon}}{\epsilon} $, 
\[ \Pr_{x,y^1, \ldots, y^k}[E_x \mid x \in X^{II}] \leq (1-\epsilon)^k \leq \frac{\epsilon^2}{2},\]
showing that 
\[  \Pr_S[I_S^b[f]<1-\epsilon] \le \frac{1}{(\epsilon/2)} \cdot \frac{\epsilon^2}{2} \le \epsilon. \qedhere \]
\end{proof}

\subsection{Small Bias Case}\label{sec:singleround_smallbias}
To handle the small bias case for the sake of proving  Theorem~\ref{thm:singleRoundSimplified}, one can simply repeat the argument of  Kahn et al.~\cite{KKL}, i.e.\ iteratively select  influential variables  and set them to the value that increases the probability of success. However, for the purposes of our results in the multi-round setting, we will need to prove a stronger result, which  states  that even if  the coalition is selected  \emph{randomly}, there is a nontrivial chance of succeeding in influencing the outcome.

\if\lipics0
 We start with the case that $\Pr[x_i=1] \in (1/4, 3/4)$ for every $i$.  The following lemma is proved in \cite{RSZ02} for the uniform distribution. By inspection, it is easy to check that the proof extends to any product distribution in which the marginal biases are bounded away from $0$ and $1$. 

\begin{lemma}[\cite{RSZ02}, Lemma 11 modified]\label{rsz}
Let $n\in \mathbb{N}$, $\gamma\in (0,1/2)$ and $m\leq n$. Assume $m > n/\gamma \log n$. Let $f\colon (\{0,1\}^n , \mu)\to \{0,1\}$, where $\mu$ is a product distribution such that  $\Ex[x_i]\in (\frac{1}{4}, \frac{3}{4})$ for each $i\in [n]$. If $\Ex[f]\geq \gamma$ then
\[
\Pr_{S\subseteq [n]\colon |S|=m} [I^1_S[f]\geq 1-\gamma]> \frac{1}{2}\left( \frac{m}{4n} \right)^{2^{\frac{80n}{m\gamma}}}.
\]	
\end{lemma}

We can extend Lemma~\ref{rsz} to somewhat higher biases by representing a $\mu_p$ distributed variable as an $\mathrm{AND}$ of $t$ variables that are $\mu_c$ distributed, where $c \approx 1/2$ and $p = c^t$.
\fi

\begin{lemma}\label{lem:single-unbiased}
Let $n\in \mathbb{N}$, $\gamma\in (0,1/2)$ and $m\leq n$. Let $f\colon (\{0,1\}^n , \mu)\to \{0,1\}$, where $\mu$ is a product distribution such that for all $i$, $1/n<\alpha \leq \Ex[x_i]\leq 1/2$. Assume $m > \frac{n\log{1/\alpha}}{2\gamma \log n}$. If $\Ex[f]\geq \gamma$ then
\[
\Pr_{S\subseteq [n]\colon |S|=m} \left[I^1_S[f]\geq 1-\gamma\right]> \frac{1}{2}\left( \frac{m}{4n\log{1/\alpha}} \right)^{2^{\frac{80n\log{1/\alpha}}{m\gamma}}}.
\]
\end{lemma}
\if\lipics0
\begin{proof}
The lemma is proved by a reduction to Lemma~\ref{rsz}. Let $\mu_i:=\Ex[x_i]$. For each variable $i\in [n]$, we pick $c_i\in (\frac{1}{4}, \frac{3}{4})$ and an integer $t_i\leq \log 1/\alpha$ such that $\mu_i=c_i^{t_i}$: first choose $t_i \leq \log_{1/4} \alpha$ so that $\left(\frac{1}{4}\right)^{t_i} < \mu_i < \left(\frac{3}{4}\right)^{t_i}$ (note the intervals are overlapping since $\left(\frac{3}{4}\right)^2 > \frac{1}{4}$), and then choose $c_i$ appropriately. For each variable $x_i$, introduce $t_i$ new variables $y_{i,1},\ldots y_{i,t_i}$. Consider $g\colon (\{0,1\}^{\sum_i t_i}, \mu')\to \{0,1\}$, where $\mu' = \prod_i \mu_{c_i}^{t_i}$ and
\[ g(y)= f\left(\bigwedge_{j=1}^{t_1} y_{1,j}, \ldots, \bigwedge_{j=1}^{t_n} y_{n,j}\right). \]
 We designed $g$ so that the input to $f$ is distributed according to $\mu$. Applying Lemma~\ref{rsz}, we deduce that typical $S$ of size $m$ satisfy $I_S^1[g] \geq 1-\gamma$. Let $S' = \{ x_i : y_{i,j} \in S \text{ for some } j\}$. A moment's thought shows that $I_{S'}^1[f] \geq 1-\gamma$. A simple coupling argument now completes the proof. 
\end{proof}
\else
\begin{proof}
The result is proved in~\cite{RSZ02}	 for constant $\alpha$. The general case follows by representing a $\mu_p$ distributed variable as an AND of $t$ variables that are distributed according $\mu_c$, where $c \approx 1/2$ and $p = c^t$. The complete details appear in the full version of the paper.
\end{proof}
\fi

\subsection{Finishing the Proof: Combining the Two Cases} \label{sec:combine} \label{sec:combined-thmsingleroundsimple}
We are ready to finish the proof of Theorem~\ref{thm:singleRoundSimplified}. Let $A:=\{i: \Pr_{x\sim \mu}[x_i=1] \in (\alpha_0, \frac{1}{2}]\}$, and recall that $\alpha_0=\frac{1}{\log n}$.  For every $y\in \{0,1\}^A$, define $f_y\colon \{0,1\}^{[n]\backslash A}\to \{0,1\}$ as $f_y(z):= f(y,z)$. By Proposition~\ref{prop:single-biased}, for every $y\in \{0,1\}^A$, there exists $b:=b_y\in \{0,1\}$ such that
\[
\Pr_{S\sim\mu_{[n] \backslash A}^{(k)}}\left[I_S^b[f_y]>1-\frac{\epsilon}{2}\right]>1-\frac{\epsilon}{2},
\]
where $k=O\left(\frac{\log(1/\epsilon)}{\epsilon}\right)$. Moreover, since every variable $i$  in $[n]\backslash A$ satisfies $\Ex[x_i]  \le \alpha_0 = 1/\log n$, Chernoff's bound gives,
\[
\Pr_{S\sim\mu_{[n] \backslash A}^{(k)}}\left[|S|\geq \frac{C\log (1/\epsilon) n}{\epsilon \log n}\right] \leq \exp \left(-\Omega\left(\frac{\log(1/\epsilon)n}{\epsilon \log n}\right)\right) \leq \frac{\epsilon}{2},
\]
for some constant $C>0$. Therefore,
\[
\Pr_{S\sim\mu_{[n] \backslash A}^{(k)}}\left[I_S^b[f_y]>1-\frac{\epsilon}{2} \text{ and } |S| \leq \frac{C\log (1/\epsilon) n}{\epsilon \log n} \right]>1-\epsilon.
\]
It follows that
\[
 \Exs_{S\sim\mu_{[n] \backslash A}^{(k)}} \left[\Pr_{y,b}\left[ I_S^b[f_y] \geq 1 - \frac{\epsilon}{2} \right]\right] > \frac{1-\epsilon}{2} \geq \frac{1}{4},
\]
assuming without loss of generality that $\epsilon \leq 1/2$. Hence, there exists a fixed $b_0\in \{0,1\}$ and a set $S$, satisfying $|S|\leq \frac{C\log (1/\epsilon) n}{\epsilon \log n}$ and  
\[
\Pr_{y}\left[ I_S^{b_0}[f_y]\geq 1- \frac{\epsilon}{2} \right] \geq \frac{1}{4}.
\]

Now, define $h\colon \{0,1\}^n\to \{0,1\}$ as $h(y)=1$ if and only if $I_S^{b_0}[f_{y\vert_A}]\geq 1-\epsilon/2$. Note that, $h$ depends only on $A$ variables. The above inequality asserts that $\Ex[h]\geq \frac{1}{4}$. Since, $A$ contains only small bias variables, we may apply Lemma~\ref{lem:single-unbiased}. Namely, there is $m=O(\frac{n\log\log n}{\epsilon \log n})$ such that
\[
\Pr_{T\subseteq [n]\colon |T|= m}\left[I_T^1[h]\geq 1-\frac{\epsilon}{2}\right] > 0.
\]
Thus, there exists a coalition $T\subseteq A$ of size $O(\frac{n \log\log n}{\epsilon \log n})$ of players that can bias $h$ towards $1$. In other words, $T$ can bias $y$ towards cases where $S$ is able to bias $f_y$ towards $b_0$. As a result, 
\[ I_{S\cup T}^{b_0}[f] \geq \left(1 - \frac{\epsilon}{2} \right)\left(1 - \frac{\epsilon}{2} \right) > 1-\epsilon. \]
Moreover, $|S\cup T| =  O\left(\frac{n \log \log n}{\epsilon \log n}+ \frac{\log(1/\epsilon) n}{\epsilon \log n}\right)$, as desired.

\section{The Larger Range: Proof of Theorem~\ref{thm:singleRound_General}}\label{sec:largerange}

As outlined in the introduction, there are certain obstacles to generalizing  Theorem~\ref{thm:singleRoundSimplified}  to larger ranges.  In particular, the fact that the set $\mathcal{E}$ of all the points on which the coalition fails in Theorem~\ref{thm:singleRoundSimplified}  is of small measure does not seem to be a sufficiently strong condition for an induction to go through. We will need to prove a strengthening of  Theorem~\ref{thm:singleRoundSimplified} which shows that not only is $\mathcal{E}$ of small measure, but it is also small if it is measured via the boosted distributions introduced in Definition~\ref{def:boosted}.  This leads to a more general definition of influence.

\begin{definition}[Boosted influence towards value]
Let $\mathcal{R}$ be an arbitrary set. For a function $f\colon \{0,1\}^n \to \mathcal{R}$ and  $b \in \mathcal{R}$, define
\if\lipics1
$I_{S}^{b,t}(f) = \Pr_{x \sim \mu^{(t)}}[ b \in f(B_S(x))]$.

\else
\[ I_{S}^{b,t}(f) = \Pr_{x \sim \mu^{(t)}}[ b \in f(B_S(x))]. \]
\fi
Note that  $I_{S}^b(f)= I_{S}^{b,1}(f)$, as $\mu^{(1)}=\mu$. 
\end{definition}

The following lemma generalizes Lemma~\ref{lem:single-biased}, as we spell out in its corollary.

\begin{lemma}\label{lem:typeI-II}
Consider $f\colon \{0,1\}^n \to  \mathcal{R}$, let $t \in \mathbb{N}$, and let $S \sim \mu^{(k)}$, where $k =\frac{10 t}{\delta}  \log \frac{ t}{\epsilon}$. Let $b \in \mathcal{R}$. We have
\if\lipics1
$\Pr_S[\forall \ell \le t, \ I_{S}^{b,\ell}(f) \ge 1-  \epsilon] \ge  1 - \epsilon$,
\else
\[ \Pr_S[\forall \ell \le t, \ I_{S}^{b,\ell}(f) \ge 1-  \epsilon] \ge  1 - \epsilon, \]
\fi
if any of the following two cases hold: 
\begin{itemize}
\item Case I: For some $s \le t$, 
\if\lipics1
$\Pr_{u \sim \mu^{(s)}}[\Pr_{v \sim \mu^{(t)}}[f(u \vee v)=b] \ge 1- \epsilon/2]  \ge \delta$.
\else
\[ \Pr_{u \sim \mu^{(s)}}[\Pr_{v \sim \mu^{(t)}}[f(u \vee v)=b] \ge 1- \epsilon/2]  \ge \delta. \]
\fi

\item Case II:  For every $s \le t$, 
\if\lipics1
$\Pr_{u \sim \mu^{(s)}} [\Pr_{v \sim \mu^{(t)}}[f(u \vee v)=b] \ge  \delta] \ge 1-\frac{\epsilon}{2}$.
\else
\[ \Pr_{u \sim \mu^{(s)}} [\Pr_{v \sim \mu^{(t)}}[f(u \vee v)=b] \ge  \delta] \ge 1-\frac{\epsilon}{2}. \]
\fi
\end{itemize}
\end{lemma}

\if\lipics1
\begin{proof}
The complete proof can be found in the full version of the paper.	
\end{proof}
\else 
\begin{proof}

\textbf{Case I:}
Suppose the condition in Case I is satisfied. Fix a $u$, and  consider an $\ell \le t$. Since  
\[ \Pr_{v \sim \mu^{(t)}}[f(u \vee v)\neq b ] = \Exs_{w \sim \mu^{(t-\ell)}} \left[\Pr_{v \sim \mu^{(\ell)}}[f(u \vee v \vee w)\neq b ] \right], \]
 by Markov's inequality $\Pr_{v \sim \mu^{(t)}}[f(u \vee v)\neq b ] \le \epsilon/2$ would imply that 
\[ \Pr_{w \sim \mu^{(t-\ell)}} \left[\Pr_{v \sim \mu^{(\ell)}}[f(u \vee v \vee w)\neq b ] \ge \epsilon \right] \le 1/2, \]
or equivalently 
\[ \Pr_{w \sim \mu^{(t-\ell)}} \left[\Pr_{v \sim \mu^{(\ell)}}[f(u \vee v \vee w)\neq b] < \epsilon \right] \ge 1/2. \]
In other words, 
\[ \Pr_{v \sim \mu^{(t)}}[f(u \vee v) = b ] \ge 1- \epsilon/2 \Longrightarrow \Pr_{w \sim \mu^{(t-\ell)}} \left[\Pr_{v \sim \mu^{(\ell)}}[f(u \vee v \vee w)=b  ] > 1- \epsilon \right] \ge 1/2. \]
 Averaging over $u$, we conclude that  by the assumption of Case I, we have  
\[ \delta/2 \le   \Pr_{\substack{u \sim \mu^{(s)} \\ w \sim \mu^{(t-\ell)}}} \left[\Pr_{v \sim \mu^{(\ell)}}[f(u \vee v \vee w)= b ] >1- \epsilon \right]= \Pr_{u \sim \mu^{(s+t-\ell)}}[\Pr_{v \sim \mu^{(\ell)}}[f(u \vee v) = b] \ge 1-  \epsilon]. \]

Hence, recalling that $S \sim \mu^{(k)}$,  the probability that there exists $u \subseteq S$ such that $\Pr_{v \sim \mu^{(\ell)}} [f(u \vee v)=b] \ge 1- \epsilon$   is at least 
\[ 1 - (1-\delta/2)^{\lfloor k/(s+t-\ell) \rfloor} \ge 1 - \epsilon/t.\]
But if this event happens then $I_{S}^{b,\ell}(f) \ge 1-  \epsilon$. Hence  by the union bound,
\[ \Pr_S[\forall \ell \le t,\ I_{S}^{b,\ell}(f) \ge 1-   \epsilon] \ge1 - \epsilon. \]

\smallskip\textbf{Case II:}
Next assume that the condition in Case II is satisfied. Consider an $\ell \le t$. Define
\[ \mathcal{B}=\{u : \Pr_{v \sim \mu^{(t)}}[f(u \vee v)=b]  <  \delta\}, \]
and note that by our assumption
\[ \mu^{(\ell)}(\mathcal{B})  \le  \frac{\epsilon}{2}. \]

Since $S \sim \mu^{(k)}$, we can set $S=y^1 \vee \dots \vee y^{k/t}$, where $y^i \sim \mu^{(t)}$ are i.i.d.\ random variables.
This shows that
\[
 I_S^{b,\ell}(f) \geq \Pr_{u \sim \mu^{(\ell)}}[f(u \lor y^i) = b \text{ for some } i].
\]
Hence
\[
 \Pr_S[I_S^{b,\ell} < 1-\epsilon] \leq
 \Pr_{y^1,\ldots,y^{k/t} \sim \mu^{(t)}}[\Pr_{u \sim \mu^{(\ell)}}[f(u \lor y^i) \neq b \text{ for all } i] > \epsilon].
\]

Define now
\[
 p_{y^1, \ldots, y^{k/t}} = \Pr_{u \sim \mu^{(\ell)}}[f(u \lor y^i) \neq b \text{ for all } i \text{ and } u \notin \mathcal{B}],
\]
and notice that
\[
 \Pr_{u \sim \mu^{(\ell)}}[f(u \lor y^i) \neq b \text{ for all } i] \leq \mu^{(\ell)}(\mathcal{B}) + p_{y^1, \ldots, y^{k/t}} \leq
 \frac{\epsilon}{2} + p_{y_1, \ldots, y^{k/t}}.
\]
Therefore,
\begin{multline*}
 \Pr_S[I_S^{b,\ell} < 1-\epsilon] \leq \Pr_{y^1,\ldots,y^{k/t} \sim \mu^{(t)}} \left[p_{y^1,\ldots,y^{k/\ell}} > \frac{\epsilon}{2}\right] \leq \frac{2}{\epsilon} \Exs_{y^1,\ldots,y^{k/t} \sim \mu^{(t)}}[p_{y^1,\ldots,y^{k/\ell}}] = \\
 \frac{2}{\epsilon} \Pr_{\substack{u \sim \mu^{(\ell)} \\ y^1,\ldots,y^{k/t} \sim \mu^{(t)}}}[f(u \lor y^i) \neq b \text{ for all } i \text{ and } u \notin \mathcal{B}].
\end{multline*}
When $u \notin \mathcal{B}$, the probability that $f(u \lor y^i) \neq b$ is at most $1-\delta$, and so
\[
 \Pr_{\substack{u \sim \mu^{(\ell)} \\ y^1,\ldots,y^{k/t} \sim \mu^{(t)}}}[f(u \lor y^i) \neq b \text{ for all } i \text{ and } u \notin \mathcal{B}] =
 \Exs_{u \sim \mu^{(\ell)}}[1_{\overline{\mathcal{B}}} \cdot \Pr_{y \sim \mu^{(t)}}[f(u \lor y) \neq b]^{k/t}] \leq (1-\delta)^{k/t} \leq \frac{\epsilon^2}{2t}.
\]
This shows that $\Pr_S[I_S^{b,\ell} < 1-\epsilon] \leq \epsilon/t$. We complete the proof by an application of the union bound.
\end{proof}
\fi

\if\lipics0
\begin{corollary}
 Consider $f\colon \{0,1\}^n \to \{0,1\}$, let $t \in \mathbb{N}$, and $S \sim \mu^{(k)}$, where $k =\frac{20 t}{\epsilon} \log \frac{ t}{\epsilon}$. At least for one of the values $b \in \{0,1\}$, we have
\[ \Pr_S[\forall \ell \le t, \ I_{S}^{b,\ell}(f) \ge 1-\epsilon] \ge  1 - \epsilon. \]
 \end{corollary}
 \begin{proof}
 Setting $\delta=\epsilon/2$, either Case I holds for $b=0$ or Case II holds for $b=1$. 
 \end{proof}
 
 Another corollary allows the function $f$ to attain a third value $\dagger$, as long as its probability is small enough (with respect to various boostings of $\mu$).
 
 \begin{corollary}
 \label{cor:baseCase}
 Consider $f\colon \{0,1\}^n \to \{0,1,\dagger\}$, let $t \in \mathbb{N}$, and let $S \sim \mu^{(k)}$, where $k=k(t,\epsilon) =\frac{40 t}{\epsilon} \log \frac{ t}{\epsilon}$.  If $\Pr_{x \sim \mu^{(\ell)}}[f(x)=\dagger] <  (\epsilon/4)^2$ for all $\ell \le 2t$, then at least for one of the values $b \in \{0,1\}$, we have
\[ \Pr_S[\forall \ell \le t, \ I_{S}^{b,\ell}(f) \ge 1-  \epsilon] \ge  1 - \epsilon. \]
 \end{corollary}
\begin{proof}
 Set $\delta=\epsilon/4$, and fix an $s$. Suppose that  neither Case I  holds for $b=0$,  nor Case II holds for $b=1$. That is  
 \begin{equation}
 \label{eq:caseI_fail}
 \Pr_{u \sim \mu^{(s)}}[\Pr_{v \sim \mu^{(t)}}[f(u \vee v)=0] \ge 1- \epsilon/2] <\epsilon/4,
 \end{equation}
and
\begin{equation}
  \label{eq:caseII_fail}
 \Pr_{u \sim \mu^{(s)}}[\Pr_{v \sim \mu^{(t)}}[f(u \vee v)=1] \ge \epsilon/4] < 1 - \epsilon/2.
 \end{equation}
 Let 
 \[ \mathcal{B}=\left\{u:  \Pr_{v \sim \mu^{(t)}}[f(u \vee v)=0] \ge 1- \epsilon/2 \ \mbox{or} \Pr_{v \sim \mu^{(t)}}[f(u \vee v)=1] \ge \epsilon/4\right\}, \]
 and note that for every $u \not\in \mathcal{B}$ we have  
 \[ \Pr_{v \sim \mu^{(t)}}[f(u \vee v)=\dagger] \ge \epsilon/4. \]
  On the other hand,  by (\ref{eq:caseI_fail}) and (\ref{eq:caseII_fail}) we have
 \[ \Pr_{u \sim \mu^{(s)}}[u \in \mathcal{B}] < 1 - \frac{\epsilon}{4}. \]
 This is a contradiction, as it implies that   
     \[ \Pr_{x \sim \mu^{(s+t)}}[f(x)=\dagger] = \Pr_{\substack{u \sim \mu^{(s)} \\ v \sim \mu^{(t)}}}[f(u \vee v)=\dagger] \ge \Pr_{u \sim \mu^{(s)}}[u \not \in \mathcal{B}] \Pr_{\substack{u \sim \mu^{(s)} \\ v \sim \mu^{(t)}}}[f(u \vee v)=\dagger| u \not\in \mathcal{B}] \ge  \left(\frac{\epsilon}{4}\right)^2 . \]
Hence for every $s$, either Case I holds for $b=0$, or Case II holds for $b=1$, and thus Lemma~\ref{lem:typeI-II} implies the corollary.  
 \end{proof}
 \fi
 
 \if\lipics1
 We can now state the main result of this section.
 \else
 We can now state and prove the main result of this section, which generalizes Corollary~\ref{cor:baseCase} to allow more output bits.
 \fi The failure output $\dagger$ allows the inductive proof of Theorem~\ref{thm:strong}, as well as our multi-round result, Theorem~\ref{thm:multiRoundDetailed}, to go through, as we explain in Section~\ref{sec:multi-round}.

\begin{theorem}\label{thm:strong}
Let $f\colon \{0,1\}^n \to \{0,1\}^m \cup \{\dagger\}$,  and suppose that $\{0,1\}^n$ is endowed with a  probability measure $\mu$. Let $t$ be a positive integer, and let $S \sim \mu^{(k)}$, where $k=k(m,t,\epsilon) = O(tm^3\epsilon^{-2} \log \frac{tm}{\epsilon})$. If $\Pr_{\mu^{(\ell)}}[\dagger] < \frac{\epsilon^4}{2^{16}}$ for every $\ell \le 2t$, then  
there exists a value $b \in \{0,1\}^m$ such that
\if\lipics1
$\Pr_S \left[\forall \ell \le t, \ I_{S}^{b,\ell}(f) \ge 1 -\epsilon \right] \ge 1 - \epsilon$.
\else
\[ \Pr_S \left[\forall \ell \le t, \ I_{S}^{b,\ell}(f) \ge 1 -\epsilon \right] \ge 1 - \epsilon. \]
\fi
\end{theorem}

\if\lipics1
\begin{proof}
The complete proof can be found in the full version of the paper.	
\end{proof}
\else
\begin{proof}
We prove this by induction on $m$. The base case $m=1$ is established in Corollary~\ref{cor:baseCase}.

 We divide $f(x)$ into two parts $f_1(x)$ and $f_2(x)$ corresponding to the first bit and the following $m-1$ bits, respectively. That is,   $f_1(x)=f_2(x)=\dagger$ if $f(x)=\dagger$, and otherwise $f_1(x)$ equals  the first  bit, and $f_2(x)$ equals  the last $m-1$ bits of $f(x)$. Let  $\epsilon_1=\left(\frac{\epsilon}{8}\right)^2$, $\epsilon_2=\epsilon - \epsilon_1$,  $k_1=k(1,2t,\epsilon_1)$, and $k_2=k(m-1,t,\epsilon_2)$.

Since $\Pr_{x \sim \mu^{(\ell)}}[f(x)=\dagger] < \left(\frac{\epsilon_1}{4}\right)^2$,  applying the base  case to $f_1$ with parameters $m=1$, $2t$ and $\epsilon_1$, we find a value $b_1 \in \{0,1\}$ such that 
\[ \Pr_{S_1 \sim \mu^{(k_1)}} [\forall \ell \le 2t, \ I_{S_1}^{b_1,\ell}(f_1) \ge 1 - \epsilon_{1}] \ge 1 - \epsilon_{1}. \]

Let us call an $S_1$ \emph{good} if for all  $\ell \le 2t$  we have $I_{S_1}^{b_1,\ell}(f_1) \ge 1 - \epsilon_{1}$. For a fixed \emph{good} $S_1$, for every $x$ satisfying $b_1 \in f_1(B_{S_1}(x))$,  let $\sigma_{S_1}(x)=y$, where $y \in B_{S_1}(x)$ is an element satisfying $f_{1}(y)=b_1$. We call such values of $x$  \emph{good} with respect to $S_1$. We call the other values of $x$ \emph{bad} with respect to $S_1$. Note that if $S_1$ is good, then not only a random $x \sim \mu$ is unlikely to be bad, but the same is true   if  $x \sim \mu^{(\ell)}$ for larger values of  $\ell$ as long as $\ell \le 2t$.  More precisely, for every $\ell \le 2t$, 
\[ \Pr_{x \sim \mu^{(\ell)}}[\mbox{$x$ is bad w.r.t.\ $S_1 \ |  \  S_1$ is good}] < \epsilon_1. \]

This stronger statement is the key property that will allow us to proceed with our strong induction. 

Now we need to force  the last $m-1$ bit of $f$. Let $g_{S_1}\colon \{0,1\}^n \to \{0,1\}^{m-1} \cup \{ \dagger\}$ be defined as $g_{S_1}(x) = f_2(\sigma_{S_1}(x))$ for good values of $x$, and $g_{S_1}(x)=\dagger$ for bad values of $x$. Note that
\[ \Pr_{x \sim \mu^{(\ell)}}[\mbox{$g_{S_1}(x)=\dagger \ |  \  S_1$ is good}] < \epsilon_1 \le \left(\frac{\epsilon_2}{4} \right)^2, \]
where the last inequality can be verified easily. 

Provided that $S_1$ is good, applying the induction hypothesis to $g_{S_1}$  with $\epsilon_2$ and $k_2$, we conclude that there exists a value $b_2 \in \{0,1\}^{m-1}$ such that 
\[ \Pr_{S_2 \sim \mu^{(k_2)}} [\forall \ell \le  t, \ I_{S_2}^{b_2,\ell}(g_{S_1}) \ge 1 - \epsilon_2] \ge 1 - \epsilon_2. \]
Let us call $S_2$ \emph{good} with respect to $S_1$ if the condition in the above probability holds. 

Now suppose that $S_1$ is good, and that $S_2$ is  good with respect to $S_1$. Then for a random $x \sim \mu^\ell$, with probability at least $1-\epsilon_2$, there is a $y \in B_{S_2}(x)$ with $g_{S_1}(y)=f_2(\sigma_{S_1}(y))=b_2$. On the other hand, $f_1(\sigma_{S_1}(y))=b_1$ and $\sigma_{S_1}(y) \in B_{S_1}(y)$. This shows that $z=\sigma_{S_1}(y)$ satisfies $z \in B_{S_1 \cup S_2}(x)$ and $f(z)=(f_1(x),f_2(x))=(b_1,b_2)=:b$. Hence conditioned on $S_1$ and $S_2$ being good, we have for $S=S_1 \cup S_2$ and every $\ell \le t$, 
\[ I_{S}^{b,\ell}(f) \ge 1 - \epsilon_2. \]

We conclude that for $S=S_1 \cup S_2$, 
\begin{align*}
\Pr_S [\forall \ell \le t, \ I_{S}^{b,\ell}(f) \ge 1-  \epsilon_2] &\ge \Pr[\mbox{$(S_1$ is good) and ($S_2$ is good w.r.t.\ $S_1$)}] \\
 &\ge \Pr[\mbox{$S_1$ is good)}]  \times \Pr[\mbox{$S_2$ is good w.r.t.\ $S_1 \ | \ S_1 $ is good}] \\
& \geq (1-\epsilon_1)(1 - \epsilon_2) \geq 1 - \epsilon_1 - \epsilon_2 = 1 - \epsilon.
 \end{align*}
 

Finally, denoting by $k(m,t,\epsilon)$ 
the value $k$ such that $S \sim \mu^{(k)}$,
we get the recurrence
\[
 k(m,t,\epsilon) \leq k(1,2t,\epsilon_1) + k(m-1,t,\epsilon_2),
\]
with base case $k(1,2t,\epsilon) = O(t/\epsilon_1 \log(t/\epsilon_1))$. Accordingly, define $K(\gamma) = (t/\gamma^2) \log (t/\gamma)$.

Let us define a sequence $\gamma_1 = \epsilon$, $\gamma_{r+1} = \gamma_r(1-\gamma_r/64)$; note that $\gamma_{r+1}$ is the $\epsilon_2$ corresponding to $\epsilon = \gamma_r$. Then
\[
 k(m,t,\epsilon) = O(K(\gamma_1) + \cdots + K(\gamma_m)).
\]

Since the $\gamma$ sequence is monotone, we see that $\gamma_r \geq \gamma_1 (1-\epsilon/64)^r$. Choose a constant $C \geq 2$ so that $s_1 = C/\epsilon$ satisfies $\gamma_{s_1} \geq \epsilon/2$. The same calculation shows that $s_2 = 2C/\epsilon$ satisfies $\gamma_{s_1 + s_2} \geq \epsilon/4$, that $s_3 = 4C/\epsilon$ satisfies $\gamma_{s_1 + s_2 + s_3} \geq \epsilon/8$, and so on. In general, $\gamma_{s_1 + \dots + s_a} \geq \epsilon/2^a$. On the other hand, $s_1 + \dots + s_a = (2^a-1)C/\epsilon$. This shows that
\[
 \gamma_{s_1 + \dots + s_a} \geq \frac{\epsilon}{s_1 + \dots + s_a}.
\]
Now choose the minimal $a$ so that $s_1 + \dots + s_a \geq m$. Then either $a = 1$ or $m \geq (s_1 + \dots + s_a)/2$. In both cases, $\gamma_m \geq \epsilon/(2m)$. Therefore,
\[
 k(m,t,\epsilon) = O(mK(\gamma_m)) = O\left(tm^3\epsilon^{-2} \log \frac{tm}{\epsilon} \right). \qedhere
\]
\end{proof}
\fi

\if\lipics1
Theorem~\ref{thm:singleRound_General} follows from Theorem~\ref{thm:strong} using an argument very similar to that in Section~\ref{sec:combined-thmsingleroundsimple}, as we show in the full version of the paper.

\else

\subsection{Proof of Theorem~\ref{thm:singleRound_General}}
Finally, we show how Theorem~\ref{thm:singleRound_General} follows from Theorem~\ref{thm:strong}. Similar to the proof of Theorem~\ref{thm:singleRoundSimplified} in  Section~\ref{sec:combined-thmsingleroundsimple}, we need to combine Theorem~\ref{thm:strong} that handles the  highly biased coordinates with the KKL argument that handles the small bias coordinates.

Let $m = \lceil \log_2 |\mathcal{R}| \rceil$, and embed $\mathcal{R}$ inside $\{0,1\}^m$.
As in Section~\ref{sec:combine}, let $A = \{ i : \Pr[x_i = 1] > \alpha_0 \}$, where $\alpha_0 = \frac{1}{\log n}$, and for every $y \in \{0,1\}^A$, define $f_y\colon \{0,1\}^{[n] \setminus A} \to \{0,1\}^m \cup \{ \dagger \}$ as $f_y(z) := f(y,z)$ (note $f_y(z) \neq \dagger$ for all $y,z$). Theorem~\ref{thm:strong} (applied with $t=1$ and $\epsilon/2$) shows that for every $y \in \{0,1\}^A$ there exists $b_y \in \{0,1\}^m$ such that
\[
 \Pr_{S \sim \mu^{(k)}_{[n]\setminus A}} \left[ I_S^{b_y}(f_y) \geq 1-\frac{\epsilon}{2} \right] \geq 1-\frac{\epsilon}{2},
\]
where $k = O(m^3 \epsilon^{-2} \log \frac{m}{\epsilon})$.
An averaging argument similar to the one in Section~\ref{sec:combine} shows that there exists $b_0 \in \{0,1\}^m$ and a set $S$ of size $O(\frac{kn}{\log n})$ such that
\[
 \Pr_y\left[ I_S^{b_0}[f_y] \geq 1-\frac{\epsilon}{2} \right] \geq \frac{1}{2^{m+1}}.
\]
We now define the function $h\colon \{0,1\}^A \to \{0,1\}$ just as in Section~\ref{sec:combine}: it equals $1$ when $I_S^{b_0}[f_y] \geq 1-\epsilon/2$. Applying Lemma~\ref{lem:single-unbiased} with $\gamma = \min(\epsilon/2, 1/2^{m+1})$, we deduce the existence of a coalition $T$ of $O(\frac{n\log\log n}{\gamma \log n})$ players such that $I_T^1[h] \geq 1 - \epsilon/2$. As in Section~\ref{sec:combine}, we conclude that
\[
 I^{b_0}_{S \cup T}[f] \geq (1-\epsilon/2)(1-\epsilon/2) \geq 1-\epsilon.
\]
Finally, the size of the coalition $S \cup T$ is
\[
 |S \cup T| = O\left(\frac{m^3 \log \frac{m}{\epsilon} \cdot n}{\epsilon^2 \log n} + \frac{(1/\epsilon + 2^m)n \log\log n}{\log n} \right),
\]
as claimed.
\fi


\section{Multi-Round Protocols: Proof of Theorem~\ref{thm:multiRoundSimplified}}
\label{sec:multi-round}
%

In this section we will prove Theorem~\ref{thm:multiRoundSimplified}, showing that even in the multi-round setting, there are no protocols that are resilient against all coalitions of size $o(n)$.  As described in the introduction, here at every round, first the players who are not in the coalition broadcast their random messages, and then the players in the coalition decide and broadcast their messages in an adversarial manner. The outcome is decided by a function $f\colon (\{0,1\}^n)^r \to \{0,1\}$.

 To be more formal, let $\mu=\mu_1 \times \dots \times \mu_r$ be a product distribution over $\{0,1\}^{rn} \equiv (\{0,1\}^n)^r$, where each $\mu_i$ is a  product distribution over $\{0,1\}^n$. An \emph{$(n,r)$ coin-flipping protocol} is simply a map $f\colon (\{0,1\}^n )^r \to \{0,1\}$.  Such a protocol is executed in $r$~rounds. In  the presence of a coalition $B\subseteq [n]$ of bad players, the protocol operates as follows. In round $i$, the players in $[n]\backslash B$ select $\alpha^{i}\in \{0,1\}^{[n]\backslash B}$ according to $\mu_i\vert_{[n]\backslash B}$.  Then, the bad players $B$ choose their values depending on $\alpha^1,\ldots, \alpha^i$. Formally, an $(n,r)$-strategy for a set $B\subseteq [n]$ is a sequence $\pi=(\pi_1,\ldots,\pi_r)$ of functions where
\if\lipics1
$\pi_i\colon (\{0,1\}^{[n]\backslash B})^i \to \{0,1\}^B$.
\else
\[ \pi_i\colon (\{0,1\}^{[n]\backslash B})^i \to \{0,1\}^B.\]
\fi
The function $\pi_i$ describes the choice of bits the bad players make in the $i$-th round based on the broadcasted bits of the good players in the first $i$ rounds. 

\begin{definition}
Let $f\colon (\{0,1\}^n )^r \to \{0,1\}$ be an $(n,r)$ coin-flipping protocol, and let $\mu$ be a product distribution on $(\{0,1\}^n)^r$. Given a Boolean value $b\in \{0,1\}$, a set $B\subseteq [n]$, and an $(n,r)$-strategy $\pi$ for the bad players $B$,
\begin{itemize}
\item $I_{\pi,B}^b(f)$ is the probability that $f$ outputs $b$ given that the bad players $B$ follow $\pi$. 
\item $I_{B}^b(f):= \sup_{\pi}\{I_{\pi,B}^b(f)\}$ is the influence of $B$ on $f$ towards $b$. 
\end{itemize}
\end{definition}

Our goal is to show that there exists a coalition $B$ of size $o(n)$ such that $I_B^b(f) \ge 1-\epsilon$ for some $b \in \{0,1\}$. For the moment, let us assume that we have only two rounds, and let $f(x,y)$ denote the protocol, where $x, y\in\{0,1\}^n$ correspond to the inputs in the first and the second round respectively. Let us also denote $f_x(y) \colon = f(x,y)$.

\myparagraph{Russell et al.~\cite{RSZ02} proof of the uniform case:} 
Pick $b \in \{0,1\}$ such that $\Pr[f(x,y)=b] \ge \frac{1}{2}$. Let $\mathcal{A}$ be the set of all $x\in \{0,1\}^n$ that satisfy  $\Pr_y[f_x(y)=b] \ge \frac{1}{4}$, and note that $\Pr_x[x \in \mathcal{A}] \ge \frac{1}{4}$.   By Lemma~\ref{\rsz} of Russell et al.~\cite{RSZ02},  for every  $x \in \mathcal{A}$,  a random coalition $S$ can bias $f_x$ towards $b$, with a probability $\delta$ that is not too small. Since $S$ is chosen randomly and independently of $x$, it follows that there exists a fixed coalition $S_0$ that can bias $f_x$ for at least a $\delta$ fraction of  $x \in \mathcal{A}$, and thus for at least a $\frac{\delta}{4}$ fraction of $\{0,1\}^n$. Let $\mathcal{A}' \subseteq \mathcal{A} \subseteq \{0,1\}^n$ denote the set of such $x$.  If $x \in \mathcal{A}'$, the coalition $S_0$ is able to bias the protocol by only interfering in the second round. The set $\mathcal{A}'$ is of measure at least $\frac{\delta}{4}$, which is not too small. Thus, applying   Lemma~\ref{\rsz} again, we can  find another coalition $T_0=o(n)$ which can modify most $x$'s to fall in $\mathcal{A}'$. Now we can form the desired coalition $B= T_0 \cup S_0$: In the first round, the players in $T_0$ try to modify $x$ into an element in $\mathcal{A}'$, and if they succeed, in the second round, the players in $S_0$ interfere to change the outcome of the protocol into $b$.  This argument easily generalizes to more rounds. 

We point out that it was crucial for the above argument, that the distribution of $S$ in Proposition~\ref{prop:single-biased} is independent of $f$.

\myparagraph{What fails for the general product distributions:}
Consider $f(x,y)$ over $\mu=\mu_1 \times \mu_2$, where $\mu_1$ is \emph{highly biased}, and $\mu_2$ is the \emph{uniform distribution}. Similar to the previous paragraph, we can find a set $\mathcal{A}' \subseteq \{0,1\}^n$,  a value $b \in\{0,1\}$, and a small coalition $S_0$ such that $\Pr_{x \sim \mu_1}[x \in \mathcal{A}'] \ge \frac{\delta}{4}$, and moreover for every $x \in \mathcal{A}'$, the coalition $S_0$ is able to influence $f$ towards $b$ by  interfering only in the second round. Now, if we are to follow the argument of Russell et al., we would like to find a set $T_0$ of players to add to the coalition such that, with high probability, $T_0$ is able to modify a random $x \sim \mu_1$ into an element in $\mathcal{A}'$. We could then conclude that $B=S_0 \cup T_0$ can bias $f$ towards $b$.   

Unfortunately, Proposition~\ref{prop:single-biased}, the highly-biased counterpart of Lemma~\ref{\rsz}, only  guarantees the existence of a small coalition $T_0$ which \emph{either} modifies a random $x \sim \mu$ into being in $\mathcal{A}'$ \emph{or} modifies a random $x \sim \mu$ into  \emph{not} being in $\mathcal{A}'$; in the latter case, the coalition $T_0$ is useless. As Example~\ref{ex:pbiased} shows, this is not just a caveat of the proof of the proposition. To be more concrete, suppose $\mu_1$ is the $\frac{1}{n}$-biased distribution, and $\mathcal{A}'$ consists only of the single element $x=\vec{0}$.  Even though $\Pr[ x\sim \mathcal{A}] \ge \frac{1}{4}$, there is no coalition of size $o(n)$ which can, with high probability, modify a  random $x \sim \mu_1$ into an element in $\mathcal{A}'$. On the other hand, even a single player can modify every $x$ into an element outside $\mathcal{A}'$, but this is not helpful for our purposes, as the elements outside $\mathcal{A}'$ are the elements that $S_0$ cannot handle. 

\myparagraph{How to overcome the problem:} 
Consider the same setting as in the previous paragraph. We know that for every $x$, a random coalition $S$ of size $o(n)$ succeeds in influencing $f_x$ towards one of the outputs, with probability at least $\delta$, where $\delta$  is not too small. Instead of picking one $S_0$, we select a collection of coalitions that cover almost all $x$'s. More precisely, we find $S_1,\ldots,S_M$ and $b_1\ldots,b_M$, where $M=O_\delta(1)$, such that apart from a small  set of exceptions $\mathcal{E} \subseteq \{0,1\}^n$, every $f_x$ can be biased towards some $b_i$ using the coalition $S_i$. 

Let $h\colon \{0,1\}^n \to \{1,\ldots,M\} \cup \{\dagger\}$ be defined as follows: If $x \in \mathcal{E}$, then $h(x)=\dagger$, and otherwise $h(x)$ is equal to  some $i$ such that $S_i$ can bias $f_x$ towards $b_i$.   This brings us  to the non-Boolean range case, which was analyzed in Section~\ref{sec:largerange}. We can apply Theorem~\ref{thm:strong} to find a coalition $T$ that can influence $h$ towards one of the values in $j \in \{1,\ldots,M\}$. Now $B=T \cup S_j$ will be our desired coalition. With high probability, in the first round the players in $T$  can successfully modify a random element $x$ into an element $x'$ with $h(x')=j$, and then in the second round, the players in $S_j$ can modify $x'$ to bias the outcome towards $b_j$.  This is the main new idea used below to resolve the multi-round setting over arbitrary distributions. 

\medskip
 
 Theorem~\ref{thm:multiRoundSimplified} is a consequence of  the following more elaborate theorem which states that  for sufficiently large $n$, and $r\leq \log^\star n/5$, no $(n,r)$  protocol over an arbitrary product distribution is resilient against coalitions of $m=o(n)$ bad players. 

\begin{theorem}
\label{thm:multiRoundDetailed}
For every $\epsilon>0$, and integers  $n>0$, and  $r<\log^\star n/5$, there exists $\delta=\Omega(\frac{1}{\log(1/\epsilon)^{r} n})$, and $m=o(n)$ such that the following holds. For every  $f\colon (\{0,1\}^n)^r\rightarrow \{0,1\}$ over a product distribution $\mu$, there exists $b \in \{0,1\}$,  such that the corresponding $r$-round protocol satisfies
\if\lipics1 
$\Pr_{S\sim \nu}[ I^b_S(f)\geq 1-\epsilon]\geq \delta_r$,
\else
\[ \Pr_{S\sim \nu}[ I^b_S(f)\geq 1-\epsilon]\geq \delta_r,\]
\fi
where  $\nu$ is a distribution on $\binom{[n]}{m}$ that depends only  $\mu$ but not on $f$. To be more precise, one can take $m=O_\epsilon\left(\frac{n \cdot r \cdot 4^r}{\log^{(4r)}n}\right)=O_\epsilon\left(\frac{n (\log^\star n)^2}{\log^{(4r)}n}\right)$.
\end{theorem}
\if\lipics1
\begin{proof}
The complete proof can be found in the full version of the paper.	
\end{proof}
\else
\begin{proof}
Let $\mu=\mu_1 \times \dots \times \mu_r$.  We make two simplifying assumptions:
\begin{itemize}
\item[1.] Without loss of generality, we may assume for every $i\in [r]$ and $j\in [n]$ that $\Ex_{x\sim \mu_i}[x_j]\leq 
\frac{1}{2}$, as otherwise, we can exchange the role of $0$ and $1$ for the $j$-th variable.  

\item[2.] Let $\eta_1,\dots,\eta_r\in (0,1)$ be set later. By potentially doubling the number of rounds, and modifying the product distributions, we will assume that for every $i\in [r]$, the distribution $\mu_i$ is either highly biased or not very biased. Namely, one of the following two cases holds
\begin{itemize}
\item Highly biased with parameter $\eta_i$: For all $j\in [n]$, $\Ex_{x\sim \mu_i}[x_j]\leq \eta_i$;
\item Small biased with parameter $\eta_i$: For all $j\in [n]$, $\Ex_{x\sim \mu_i}[x_j]> \eta_i$;
\end{itemize}

In more detail, let $\hat\mu_1,\ldots,\hat\mu_{r/2}$ be the original distributions. If $\eta_{2i-1} \leq \eta_{2i}$, then for each $j \in [n]$, either $\Ex_{x \sim \hat\mu_i}[x_j] \leq \eta_{2i}$ or $\Ex_{x \sim \hat\mu_i}[x_j] > \eta_{2i-1}$. We let $\mu_{2i-1}$ be the highly biased distribution with parameter $\eta_{2i-1}$ obtained from $\hat\mu_i$ by replacing the $\eta_{2i-1}$-unbiased coordinates with dummy coordinates, and similarly $\mu_{2i}$ is the small biased distribution with parameter $\eta_{2i}$ obtained from $\hat\mu_i$ in an analogous fashion. If $\eta_{2i-1} \geq \eta_{2i}$, then instead $\mu_{2i-1}$ will be $\eta_{2i-1}$-unbiased and $\mu_{2i}$ will be $\eta_{2i}$-biased.
\end{itemize}

Let $\delta_0:=1$, $m_0:=0$, and $k_0:=0$. For every $\ell \in [r]$ we set the following parameters, some of them recursively:
\begin{itemize}
\item $\delta_\ell:= \Theta(\frac{1}{\log^{\ell}(1/\epsilon)\log^{(4r-4\ell)} n})$
\item $\eta_\ell:= \delta_{\ell-1}$
\item $k_\ell:=O(\frac{4^{r-\ell} n}{ \log^{(4r)}n\cdot \epsilon^3} )$
\item $m_\ell:= k_\ell + m_{\ell-1}$
\end{itemize}

We will show by induction on $\ell$ that the following modified statement of the theorem holds.

\begin{quotation}
Let $g\colon (\{0,1\}^n)^\ell \to \{0,1\}$ be an $(n,\ell, \mu_{r-\ell+1} \times \ldots \times \mu_r)$-protocol. There is a choice of $b\in \{0,1\}$ and a probability measure $\nu_\ell$ over subsets of $[n]$ of size at most $m_\ell$, such that 
\[ \Pr_{S\sim \nu_\ell}\left[ I^b_S(g)\geq 1-\frac{\epsilon}{2^{r-\ell}} \right]\geq \delta_\ell> 0. \]
Moreover, $\nu_\ell$ does not depend on $g$.
\end{quotation}

 Note that the $\ell=r$ case is then the statement of the theorem. The base case of $\ell=0$ is about biasing a zero-round protocol (namely, a protocol that outputs a constant value in $\{0,1\}$ with no players involved). The base case of $\ell=0$ is trivially true, as no bad players are needed to fully bias a constant valued protocol with probability $1$. 

%

For the induction step, in the case when the first round of $g$ is highly biased, we apply the following lemma.

\begin{lemma}[Large bias]\label{lem:multi-biased}
 Let $g\colon (\{0,1\}^n)^\ell \to \{0,1\}$ be an $\ell$-round coin-flipping protocol, and suppose that for each $i\in [\ell]$, the $i$-th round is endowed with the distribution $\mu_{r-\ell+i}$. Suppose that $\mu_{r-\ell+1}$ is highly biased with parameter $\eta_\ell$.
 There is $b \in \{0,1\}$ such that for \[ k=O\left(\epsilon^{-3} 4^{r-\ell}\log r \log(\frac{1}{\delta_{\ell-1}})\log\log(\frac{1}{\delta_{\ell-1}})\right) \] it holds that
\[
\Pr_{S\sim \mu_{r-\ell+1}^{(k)}, T\sim \nu_{\ell-1}}\left[ I_{S\cup T}^b(g)\geq 1-\frac{\epsilon}{2^{r-\ell}}\right] \geq \frac{\delta_{\ell-1}}{4\log(1/\epsilon)+4r}. 
\]
(Here $\mu_{r-\ell+1}^{(k)}$ is the union of the supports of $k$ independent samples from $\mu_{r-\ell+1}$.)
\end{lemma}
\begin{proof}
Throughout the proof it helps to think of sampling $T\sim \nu_{\ell-1}$ in two stages: We first sample $M=O\bigl(\frac{\log(1/\epsilon)+r}{\delta_{\ell-1}}\bigr)$ sets $T_1,\dots,T_M$ independently from $\nu_{\ell-1}$. We later sample $i\in [M]$ uniformly at random and let $T=T_i$. Note that, even though we sampled $T$ in two stages, $T$ is still distributed according to $\nu_{\ell-1}$. 

For every $x\in \{0,1\}^n$ define an $(\ell-1)$-round protocol $g_x\colon (\{0,1\}^n)^{\ell-1} \to \{0,1\}$ as $g_x(y):=g(x,y)$. By the induction hypothesis, for every $x$, 
\begin{equation}\label{eq:IH1}
\Pr_{T_1,\ldots,T_M}\left[\exists i\in [M], \  \left(I_{T_i}^1(g_x)\geq 1-\frac{\epsilon}{2^{r-\ell+1}} \; \text{ or } \; I_{T_i}^0(g_x)\geq 1-\frac{\epsilon}{2^{r-\ell+1}}\right)\right] \geq 1- (1-\delta_{\ell-1})^M. 
\end{equation}

For $i\in [M]$ let $\bin(i)\in \{0,1\}^{\lceil \log M \rceil}$ denote the binary representation of $i-1$ (we chose $i-1$ so that the all zeros vector does not go unused). We define a function $h\colon \{0,1\}^n \to \{0,1\}^{\lceil \log M\rceil + 1} \cup \{\dagger\}$ as follows. We set $h(x)=\dagger$ if for every $i\in [M]$, both $I_{T_i}^1(g_x)< 1-\frac{\epsilon}{2^{r-\ell+1}}$ and $I_{T_i}^0(g_x)< 1-\frac{\epsilon}{2^{r-\ell+1}}$. Otherwise, we let $h(x)=(\bin(i), b)$, where $(i,b)\in [M] \times \{0,1\}$ is the lexicographically first tuple such that $I_{T_i}^b(g_x)\geq 1-\frac{\epsilon}{2^{r-\ell+1}}$.

We will apply Theorem~\ref{thm:strong} to $h$ with $t=1$, but before doing so, we will show that the conditions on the $\dagger$ probability will hold with high probability over the choice of $T_1,\dots,T_M$. We need to verify that $\Pr_{\mu_{r-\ell+1}}[\dagger]\leq \frac{(\epsilon/2^{r-\ell+1})^4}{2^{16}}$ and $\Pr_{\mu_{r-\ell+1}^{(2)}}[\dagger]\leq \frac{(\epsilon/2^{r-\ell+1})^4}{2^{16}}$. We first observe that, 
\[
\Pr_{\mu_{r-\ell+1}}[\dagger]= \Pr_{x\sim \mu_{r-\ell+1}}\left[\forall i\in [M], b\in \{0,1\}, \ I_{T_i}^b[g_x]<1-\frac{\epsilon}{2^{r-\ell+1}}  \right]. 
\]
Thus (\ref{eq:IH1}) gives us,
\[
\Exs_{T_1,\ldots, T_M} \left[\Pr_{\mu_{r-\ell+1}}[\dagger]\right] \leq (1-\delta_{\ell-1})^M \leq \frac{(\epsilon/2^{r-\ell+1})^5}{2^{17}}.
\]
Applying Markov's inequality, we get
\[
\Pr_{T_1,\ldots, T_M}\left[ \Pr_{\mu_{r-\ell+1}}[\dagger]\geq  \frac{(\epsilon/2^{r-\ell+1})^4}{2^{16}} \right] \leq  \frac{(\epsilon/2^{r-\ell+1})}{2}.
\]
An identical argument gives
\[
\Pr_{T_1,\ldots, T_M}\left[ \Pr_{\mu_{r-\ell+1}^{(2)}}[\dagger]\geq  \frac{(\epsilon/2^{r-\ell+1})^4}{2^{16}} \right] \leq  \frac{(\epsilon/2^{r-\ell+1})}{2}.
\]
Thus, $h$ satisfies the conditions of Theorem~\ref{thm:strong} for $t=1$ with probability at least $1-\frac{\epsilon}{2^{r-\ell+1}}$. Define $E_T$ to be this event. Conditioned on $E_T$, there exist $i\in [M]$ and $b\in \{0,1\}$ for which
\[
\Pr_{S\sim \mu_{r-l+1}^{(k)}} \left[ I_S^{(\bin(i),b)}[h]\geq 1-\epsilon/2^{r-\ell+1} \right]\geq 1-\epsilon/2^{r-\ell+1}, 
\]
where $k=O\left(\log^{(3)} M (\frac{\epsilon}{2^{r-\ell}})^{-2} \log(\frac{\log M}{\epsilon})\right)= O\left(4^{r-\ell}\log r \cdot \epsilon^{-3} \log(\frac{1}{\delta_{\ell-1}})\log\log(\frac{1}{\delta_{\ell-1}})\right)$. 

Let $S$ be such that $I_S^{(\bin(i),b)}[h]\geq 1-(\epsilon/2^{r-\ell+1})$. This means that a set of bad players $S$ can use the first round to bias $h$ towards inputs $x$, for most of which, $T_i$ can be used to bias $g_x$ towards $b$. As a result, for such $S$, $I_{S\cup T_i}^b(g)\geq (1-\epsilon/2^{r-\ell+1})(1-\epsilon/2^{r-\ell+1})\geq 1-\epsilon/2^{r-\ell}$. The probability that the final set $T$ is equal to this specific $T_i$ is $1/M$. To sum up, we have proved
\begin{align*}
\Pr_{S\sim \mu_{r-\ell+1}^{(k)},T\sim \nu_{\ell-1}}\left[I_{S\cup T}^b[g] \geq 1-\frac{\epsilon}{2^{r-\ell}} \right]&\geq 
\Pr_{S\sim \mu_{r-\ell+1}^{(k)}} \left[ I_S^{(\mathrm{bin}(i),b)}[h]\geq 1-\frac{\epsilon}{2^{r-\ell+1}} \,\middle\vert\, E_T\right]\cdot \Pr[E_T]\cdot \Pr_{T}[T=T_i]\\ &\geq 
\frac{1}{M}\cdot \left(1-\frac{\epsilon}{2^{r-\ell+1}}\right)\cdot \Pr_{S\sim \mu_{r-\ell+1}^{(k)}} \left[ I_S^{(\mathrm{bin}(i),b)}[h]\geq 1-\frac{\epsilon}{2^{r-\ell+1}}  \,\middle\vert\, E_T\right]\\ &\geq 
\left(1-\frac{\epsilon}{2^{r-\ell+1}}\right)^2\cdot M^{-1}\\ &\geq \frac{\delta_{\ell-1}}{4\log(1/\epsilon)+4r}. \qedhere
\end{align*}
\end{proof}

Note that for $S\sim \mu_{r-\ell+1}^{(k)}$ we have $\Ex_{S}[|S|]\leq k\eta_{\ell}n$. Applying the Chernoff bound, choosing $\eta_\ell\leq \frac{1}{n}$ we have
\[
Pr_{S\sim \mu_{r-\ell+1}^{(k)}}[|S|\geq  2k\eta_\ell n] \leq e^{\frac{-(k\eta_\ell n)^2}{kn\eta_\ell(1-\eta_\ell)}} \leq \frac{\delta_{\ell-1}}{8\log(1/\epsilon)+8r},
\]
which follows from our choice of $k$.

We let $\nu_\ell$ be the distribution that samples $S\sim \mu_{r-\ell+1}^{(k)}$ conditioned on $|S|\leq 2k\eta_\ell n$ and $T\sim \nu_{\ell-1}$ and takes their union. It follows from the above Chernoff bound and Lemma~\ref{lem:multi-biased} that there exists $b\in \{0,1\}$ such that
\[
\Pr_{S\sim \nu_{\ell}}\bigl[I^b_S[g]\geq 1-\frac{\epsilon}{2^{r-\ell}}\bigr]\geq \frac{\delta_{\ell-1}}{8\log(1/\epsilon)+8r}.
\]
In particular, if  
\begin{equation}\label{cond1}
\delta_\ell\leq \frac{\delta_{\ell-1}}{8\log(1/\epsilon)+8r}	
\end{equation}
and 
\begin{equation}\label{cond2}
m_\ell\geq m_{\ell-1}+ \epsilon^{-3}\log(1/\delta_{\ell-1})\log\log(\delta_{\ell-1})4^{r-\ell}\log r\cdot \eta_{\ell}n,	
\end{equation}
the distribution $\nu_\ell$ satisfies the induction step. 

We now verify these two conditions by recalling our choices of $\delta_\ell:= \Theta(\frac{1}{\log(1/\epsilon)^\ell\log^{(4r-4\ell)} n})$, $\eta_\ell:= \delta_{\ell-1}$,  $k_\ell:=O(\frac{ 4^{r-\ell} n}{ \log^{(4r)}n\cdot \epsilon^3} )$, and $m_\ell:= \ell \cdot k_\ell$. Now, (\ref{cond1}) follows from 
\[
\delta_{\ell+1} \leq \frac{\delta_{\ell}}{8\log(1/\epsilon)+ 8\log r},
\]
and (\ref{cond2}) holds because
\[
k_{\ell+1} \geq \epsilon^{-3}\cdot \log(\frac{1}{\delta_{\ell-1}}) \log\log(\delta_{\ell-1}) 4^{r-\ell+1}\log r\cdot \eta_{\ell} n. 
\]

In the case when the first round of $g$ has small bias, we apply the following lemma for the induction step. 

\begin{lemma}[Small bias]\label{lem:multi-unbiased}
Let $g\colon (\{0,1\}^n)^\ell \to \{0,1\}$ be an $\ell$-round coin-flipping protocol, and suppose that for each $i\in [\ell]$, the $i$-th round is endowed with the distribution $\mu_{r-\ell+i}$. Suppose that $\mu_{r-\ell+1}$ is small biased with parameter $\eta_\ell$. There is $b\in \{0,1\}$ such that for $k>\frac{n \log(1/\eta_{\ell})}{\delta_{\ell-1} \log n}$ it holds that
\[
\Pr_{\substack{S\subseteq [n]; |S|=k \\ T\sim \nu_{\ell-1}}}\left[I_{S\cup T}^b(g)\geq 1-\frac{\epsilon}{2^{r-\ell}} \right] \geq \frac{\delta_{\ell-1}}{2} \cdot \left(\frac{k}{2n \log(1/\eta_\ell)}\right)^{2^{\frac{320\cdot 2^{r-\ell}\cdot n\cdot \log(1/\eta_\ell)}{k\cdot \epsilon\cdot \delta_{\ell-1}}}}. 
\]
\end{lemma}
\begin{proof}
Similar to the proof of 	Lemma~\ref{lem:multi-biased}, we consider the functions $g_x(y):=g(x,y)$. 
Let $E_x$ be the event that at least one of $I_T^1(g_x)\geq 1-\epsilon/2^{r-\ell+1}$ or $I_T^0(g_x)\geq 1-\epsilon/2^{r-\ell+1}$ holds. By the induction hypothesis, for every $x\in \{0,1\}^n$,
\[
\Pr_{T\sim \nu_{\ell-1}}[E_x] \geq \delta_{\ell-1}. 
\]
Thus 
\[
\Exs_{T\sim \nu_{\ell-1}}\left[\Exs_{x\sim \mu_{r-\ell+1}}[1_{E_x}]\right] \geq \delta_{\ell-1}. 
\]
On the other hand,
\[
\Exs_{T\sim \nu_{\ell-1}}\left[\Exs_{x\sim \mu_{r-\ell+1}}[1_{E_x}]\right] \leq \frac{\delta_{\ell-1}}{2}+ \Pr_{T\sim \nu_{\ell-1}}\left[\Exs_{x\sim \mu_{r-\ell+1}}[1_{E_x}]\geq \frac{\delta_{\ell-1}}{2}\right],
\]
and thus 
\[
\Pr_{T\sim \nu_{\ell-1}}\left[\Pr_{x\sim \mu_{r-\ell+1}}[E_x]\geq \frac{\delta_{\ell-1}}{2}\right] \geq \frac{\delta_{\ell-1}}{2}.
\]
Recalling the definition of $E_x$, there exists $b\in \{0,1\}$ such that 
\[
\Pr_{T\sim \nu_{\ell-1}}\left[\Pr_{x\sim \mu_{r-\ell+1}}\left[I_T^b(g_x)\geq 1-\frac{\epsilon}{2^{r-\ell+1}}\right]\geq \frac{\delta_{\ell-1}}{4}\right] \geq \frac{\delta_{\ell-1}}{2}.
\]
Let $E_T$ denote the event that $\Pr_{x\sim \mu_{r-\ell+1}}\bigl[I_T^b(g_x)\geq 1-\frac{\epsilon}{2^{r-\ell+1}}\bigr]\geq \frac{\delta_{\ell-1}}{4}$. For a fixed $T$, define $h\colon \{0,1\}^n\to \{0,1\}$ as $h(x)=1$ if and only if $I_T^b(g_x)\geq 1-\frac{\epsilon}{2^{r-\ell+1}}$. 
Note that $\Ex_{x}[h(x)\vert E_T]\geq \delta_{\ell-1}/2$. Hence, by Lemma~\ref{lem:single-unbiased}, for $k>\frac{n \log(1/\eta_\ell)}{\delta_{\ell-1} \log n}$, assuming $E_T$ we have
\[
\Pr_{S\subseteq [n]; |S|=k}\left[I_S^{1}[h]\geq 1-\frac{\epsilon}{2^{r-\ell+1}}\right]\geq \frac{1}{2}\left(\frac{k}{4n\log(1/\eta_\ell)}\right)^{2^{\frac{320\cdot 2^{r-\ell}\cdot n\cdot \log(1/\eta_\ell)}{k\cdot \epsilon\cdot \delta_{\ell-1}}}}
\]
Now, note that whenever $h(x)=1$, the set $T$ is able to use the $y$ bits to bias $g(x,\cdot)$ towards $b$. In particular, we have shown
\begin{align*}
\Pr_{\substack{S\subseteq [n]; |S|=k \\ T\sim \nu_{\ell-1}}}\left[I_{S\cup T}^b(g)\geq 1-\frac{\epsilon}{2^{r-\ell}}\right] &\geq \Pr_{\substack{S\subseteq [n]; |S|=k \\ T\sim \nu_{\ell-1}}}\left[ E_T \wedge I_S^1[h]\geq 1-\frac{\epsilon}{2^{r-\ell+1}}\right]
\\ &= \Pr_{T\sim \nu_{\ell-1}}[E_T]\cdot \Pr_{S\subseteq [n]; |S|=k}\left[I_S^1(h)\geq 1-\frac{\epsilon}{2^{r-\ell+1}} \,\middle\vert\, E_T \right] \\ &\geq
\frac{\delta_{\ell-1}}{2} \cdot \left(\frac{k}{2n \log(1/\eta_\ell)}\right)^{2^{\frac{320\cdot 2^{r-\ell}\cdot n\cdot \log(1/\eta_\ell)}{k\cdot \epsilon\cdot \delta_{\ell-1}}}}. \qedhere
\end{align*}
\end{proof}

Hence, if $r-\ell+1$ is even, applying the above lemma provides the induction step as long as 
\begin{equation}\label{cond3}
\delta_\ell \leq \frac{\delta_{\ell-1}}{2} \cdot \left(\frac{k_\ell}{2n \log(1/\eta_\ell)}\right)^{2^{\frac{320\cdot 2^{r-\ell}\cdot n\cdot \log(1/\eta_\ell)}{k\cdot \epsilon\cdot \delta_{\ell-1}}}},
\end{equation}
and
\begin{equation}\label{cond4}
m_\ell \geq k_\ell+m_{\ell-1}.
\end{equation}


Recall, again, our choices of $\delta_\ell:= \Theta(\frac{1}{\log(1/\epsilon)^\ell\log^{(4r-4\ell)} n})$, $\eta_\ell:= \delta_{\ell-1}$,  $k_\ell:=O(\frac{4^{r-\ell} n}{ \log^{(4r)}n\cdot \epsilon^3} )$, and $m_\ell:= m_{\ell-1} + k_\ell$. We see that for sufficiently large $n$, and every $\ell<r-1$,
\begin{align*}
\frac{\delta_{\ell}}{2} \cdot \left(\frac{k_\ell}{2n \log(1/\eta_\ell)}\right)^{2^{\frac{320\cdot 2^{r-\ell}\cdot n\cdot \log(1/\eta_\ell)}{k_\ell\cdot \epsilon\cdot \delta_{\ell-1}}}} &\geq 
\frac{\delta_\ell}{2} \left( \frac{160 \cdot 2^{r-\ell}}{\log^{(4r-4\ell+1)}n \cdot \log^{(4r)}n} \right)^{2^{\log^{(4r-4\ell)}n \cdot \log^{(4r-4\ell+1)}n\cdot \log^{(4r)}n}}\\  &\geq 
\frac{\delta_\ell}{2} \left( \frac{160 \cdot 2^{r-\ell}}{\log^{(4r-4\ell)}n } \right)^{2^{\log^{(4r-4\ell-1)}n}} \\ &\geq
\frac{\delta_\ell}{2} \cdot 2^{\log^{(4r -4\ell-2)}n} \\ &\geq \frac{\delta_\ell}{2} \cdot \frac{1}{\log^{(4r-4\ell -3)}n} 
\\ &\geq \delta_{\ell+1},
\end{align*}
implying (\ref{cond3}). Finally (\ref{cond4}), namely $m_{\ell+1}\geq m_\ell + k_{\ell+1}$ is immediate from our definition of $m_{\ell+1}$.
\end{proof}
\fi

\section{Concluding Remarks and Open Problems} \label{sec:conclusion}

\begin{itemize}
\item  
Perhaps the most interesting next step is proving limitations for resilience of protocols where players may send longer messages. As was discussed below Conjecture~\ref{conj:Friedgut}, it is conjectured that even when the players are allowed to broadcast arbitrarily long messages, only resilience against coalitions of size $o(n)$ is possible. This question has also been studied in the multi-round setting~\cite{RSZ02, RZ01, Feige99}. In this case, if the players are allowed $\log n$-bit messages, we know of $\left(\log^\star n+O(1)\right)$-round protocols resilient against coalitions of size $(1/2-\epsilon) n$~\cite{RZ01, Feige99}. On the other hand, Russell et al.~\cite{RSZ02}  showed that $\Omega(\log^\star n)$ rounds are necessary if we have the added restriction that in the $i$-th round the players are allowed messages of length $(\log^{(2i-1)} n)^{1-o(1)}$. Strengthening this impossibility result to messages of length $\Omega(\log n)$ is another interesting problem that remains open. 

\item 
The key qualitative point of Theorems~\ref{thm:singleRoundSimplified} and~\ref{thm:multiRoundSimplified} is that there always exists a coalition of size $o(n)$ that can bias the outcome of the protocol towards a particular value. Interestingly, we are not aware of a simpler proof of this weaker qualitative statement even in the case of the uniform measure.  The proof techniques introduced in this paper for the highly  biased coordinates are more combinatorial and probabilistic in nature; however, the less biased coordinates are ultimately handled by the Fourier-analytic proof of~\cite{KKL}. These Fourier analytic arguments are  \emph{hard} in nature, in the sense that their purpose is   to give effective bounds. It would be interesting to find more intuitive combinatorial proofs for these statements, potentially at the cost of obtaining less effective bounds, or by appealing to \emph{soft analytic tools} such as compactness, at the cost of obtaining no quantitative  bounds. We refer the reader to Terence Tao's blog post~\cite{Tao07} for a discussion about hard and soft analysis. 

\item 
Over the uniform distribution, Kahn et al.~\cite{KKL} proved that there exists no Boolean function that is $\epsilon$-resilient against coalitions of size $\omega_\epsilon\left(\frac{n}{\log n}\right)$. In this work we show that a similar bound of $\omega_\epsilon\left( \frac{n \log \log n}{\log n}\right)$ on resilience holds over arbitrary product distributions. A natural question is whether the $\log\log n$ in our bound necessary. However, even in the uniform setting there is work left to be done. Here, the best known constructions guarantee resilience against coalitions of size $O(\frac{n}{\log^2 n})$~\cite{Meka17, AL93}, which is a factor of $\log n$ off from the impossibility result of Kahn, Kalai, and Linial. 
\end{itemize}


\bibliography{mybib}
\end{document}